\documentclass[a4paper]{article}
\usepackage{fullpage}

\usepackage[T1]{fontenc}
\usepackage[utf8]{inputenc}
\usepackage{amsmath,amsfonts,amssymb,amsthm,amscd}
\usepackage{authblk}
\usepackage{algorithm}
\usepackage[demo]{graphicx}
\usepackage{caption}
\usepackage{subcaption}
\usepackage{algpseudocode}
\usepackage{forest}
\usepackage{tikz}

\usepackage{dsfont}
\usepackage{graphicx}
\usepackage{xspace}
\usepackage{enumitem}

\usepackage[colorinlistoftodos]{todonotes}

\usepackage[unicode,colorlinks,allcolors=blue]{hyperref}   
\usepackage{aliascnt}
\usepackage{cleveref}

\numberwithin{equation}{section}

\newtheorem{theorem}{Theorem}[section]
\newtheorem*{theorem*}{Theorem}
\newtheorem*{lemma*}{Lemma}

\newaliascnt{lemma}{theorem}
\newtheorem{lemma}[lemma]{Lemma}
\aliascntresetthe{lemma}
\crefname{lemma}{Lemma}{Lemmas}

\newaliascnt{claim}{theorem}

\aliascntresetthe{claim}
\crefname{claim}{Claim}{Claims}

\newaliascnt{corollary}{theorem}
\newtheorem{corollary}[corollary]{Corollary}
\aliascntresetthe{corollary}
\crefname{corollary}{Corollary}{Corollaries}

\newaliascnt{construction}{theorem}

\aliascntresetthe{construction}
\crefname{construction}{Construction}{Constructions}

\newaliascnt{fact}{theorem}

\aliascntresetthe{fact}
\crefname{fact}{Fact}{Facts}

\newaliascnt{proposition}{theorem}

\aliascntresetthe{proposition}
\crefname{proposition}{Proposition}{Propositions}

\newaliascnt{conjecture}{theorem}

\aliascntresetthe{conjecture}
\crefname{conjecture}{Conjecture}{Conjectures}

\newaliascnt{definition}{theorem}
\newtheorem{definition}[definition]{Definition}
\aliascntresetthe{definition}
\crefname{definition}{Definition}{Definitions}

\newaliascnt{remark}{theorem}

\aliascntresetthe{remark}
\crefname{remark}{Remark}{Remarks}

\newaliascnt{observation}{theorem}

\aliascntresetthe{observation}
\crefname{observation}{Observation}{Observations}

\newaliascnt{notation}{theorem}

\aliascntresetthe{notation}
\crefname{notation}{Notation}{Notations}

\newcommand{\ie}{i.e.,\ }

\newcommand{\B}{\{0,1\}}

\newcommand{\Z}{\mathbb{Z}}

\newcommand{\fo}{f_0}
\newcommand{\fl}{f_1}
\newcommand{\claw}[1]{h_{#1}}

\newcommand{\problem}[1]{\ensuremath{\textsc{#1}}\xspace}

\newcommand{\Pigeon}{\problem{Pigeon}}
\newcommand{\Collision}{\problem{Collision}}
\newcommand{\Prefixcollision}{\problem{Prefix-Collision}}
\newcommand{\Dove}{\problem{Dove}}
\newcommand{\Claw}{\problem{Claw}}
\newcommand{\GeneralClaw}{\problem{General-Claw}}

\newcommand{\DLog}{\problem{DLog}}
\newcommand{\DLogp}{\ensuremath{\DLog_p}\xspace}
\newcommand{\Index}{\problem{Index}}
\newcommand{\Blichfeldt}{\problem{Blichfeldt}}

\newcommand{\class}[1]{\ensuremath{\mathsf{#1}}\xspace}
\newcommand{\TFNP}{\class{TFNP}}
\newcommand{\TFUP}{\class{TFUP}}

\newcommand{\PPAD}{\class{PPAD}}
\newcommand{\PPA}{\class{PPA}}
\newcommand{\PPP}{\class{PPP}}
\newcommand{\PWPP}{\class{PWPP}}
\newcommand{\NP}{\class{NP}}

\newcommand{\DLP}{DLP\xspace}

\newcommand{\N}{\mathbb{N}}
\newcommand{\GG}{\mathds{G}}

\newcommand{\ind}{\mathcal{I}_\GG}

\newcommand{\htodo}[1]{}
\newcommand{\ptodo}[1]{}

\newcommand{\circuit}[1]{\textsf{#1}}
\newcommand{\bd}{\textup{bd}}
\newcommand{\bc}{\textup{bc}}
\newcommand{\I}{\mathcal{I}_\mathds{G}}
\newcommand{\m}[1]{\textbf{#1}}

\begin{document}

\title{On Search Complexity of Discrete Logarithm\thanks{A preliminary version of this work appeared in the 46th International Symposium on Mathematical Foundations of Computer Science, MFCS 2021~\cite{HubacekV21}.
Research was supported by the Grant Agency of the Czech Republic under the grant agreement no. 19-27871X and by the Charles University projects PRIMUS/17/SCI/9 and UNCE/SCI/004.}}
\author[1]{Pavel Hubáček}
\author[2]{Jan Václavek}
\affil[1,2]{Charles University}
\affil[1]{\texttt{hubacek@iuuk.mff.cuni.cz}}
\affil[2]{\texttt{honza.vaclavek@email.cz}}
\date{}

\maketitle

\begin{abstract}
In this work, we study the discrete logarithm problem in the context of \TFNP~-- the complexity class of search problems with a syntactically guaranteed existence of a~solution for all instances.
Our main results establish that suitable variants of the discrete logarithm problem are complete for the complexity class \PPP, respectively \PWPP, i.e., the subclasses of \TFNP capturing total search problems with a solution guaranteed by the pigeonhole principle, respectively the weak pigeonhole principle.
Besides answering an open problem from the recent work of Sotiraki, Zampetakis, and Zirdelis (FOCS’18), our completeness results for \PPP and \PWPP have implications for the recent line of work proving conditional lower bounds for problems in \TFNP under cryptographic assumptions.
In particular, they highlight that any attempt at basing average-case hardness in subclasses of \TFNP (other than \PWPP and \PPP) on the average-case hardness of the discrete logarithm problem must exploit its structural properties beyond what is necessary for constructions of collision-resistant hash functions.

Additionally, our reductions provide new structural insights into the class \PWPP by establishing two new \PWPP-complete problems.
First, the problem \Dove, a relaxation of the \PPP-complete problem \Pigeon. \Dove is the first \PWPP-complete problem not defined in terms of an explicitly shrinking function.
Second, the problem \Claw, a total search problem capturing the computational complexity of breaking claw-free permutations.
In the context of \TFNP, the \PWPP-completeness of \Claw matches the known intrinsic relationship between collision-resistant hash functions and claw-free permutations established in the cryptographic literature.
\end{abstract}

\newpage

\section{Introduction}\label{sec:Introduction}

The Discrete Logarithm Problem (DLP) and, in particular, its conjectured average-case hardness lies at the foundation of many practical schemes in modern cryptography.
To day, no significant progress towards a generic efficient algorithm solving \DLP has been made (see, e.g., the survey by Joux, Odlyzko, and Pierrot~\cite{JouxOP14} and the references therein).

One of the distinctive properties of \DLP is its \emph{totality}, i.e., given a generator $g$ of a cyclic multiplicative group $(\GG,\star)$, we know that a solution $x$ for \DLP exists for any target element $t=g^x$ in the group.
Thus, the perceived hardness of \DLP does not stem from the uncertainty whether a solution exists but pertains to the search problem itself.
In this respect, \DLP is not unique -- there are various total search problems with unresolved computational complexity in many domains such as algorithmic game theory, computational number theory, and combinatorial optimization, to name but a few.
More generally, the complexity of all total search problems is captured by the complexity class \TFNP.

In order to improve our understanding of the seemingly disparate problems in \TFNP, Papadimitriou~\cite{Papadimitriou94} suggested to classify total search problems based on syntactic arguments ensuring the existence of a solution.
His approach proved to be extremely fruitful and it gave rise to various subclasses of \TFNP that cluster many important total search problems.
For example,
\begin{description}
\item[\PPAD:] formalizes parity arguments on directed graphs and captures, e.g., the complexity of computing (approximate) Nash equilibria in normal-form games~\cite{DaskalakisGP09,ChenDT09}.
\item[\PPA:] formalizes parity arguments on \emph{undirected} graphs and captures, e.g., the complexity of Necklace splitting~\cite{Filos-RatsikasG19}.
\item[\PPP:] formalizes the pigeonhole principle and captures, e.g., the complexity of solving problems related to integer lattices~\cite{SotirakiZZ18}.
\item[\PWPP:] formalizes the \emph{weak} pigeonhole principle and captures, e.g., the complexity of breaking collision-resistant hash functions and solving problems related to integer lattices~\cite{SotirakiZZ18}.
\end{description}

\paragraph{\DLP and \TFNP.}
\DLP seems to naturally fit the \TFNP landscape.
Though, a closer look reveals a subtle issue regarding its totality stemming from the need to certify that the given element $g$ is indeed a generator of the considered group $(\GG,\star)$ or, alternatively, that the target element $t$ lies in the subgroup of $(\GG,\star)$ generated by $g$.
If the order $s=|\GG|$ of the group $(\GG,\star)$ is known then there are two natural approaches.
The straightforward approach would be to simply allow additional solutions in the form of distinct $x,y\in[s]=\{0,\ldots,s-1\}$ such that $g^x=g^y$.
By the pigeonhole principle, either $t=g^x$ for some $x\in[s]$ or there exists such a non-trivial collision $x,y\in[s]$.
The other approach would be to leverage the Lagrange theorem that guarantees that the order of any subgroup must divide the order of the group itself.
If we make the factorization of the order $s$ of the group a part of the instance then it can be efficiently tested whether $g$ is indeed a generator.

Despite being a prominent total search problem, \DLP was not extensively studied in the context of \TFNP so far.
Only recently, Sotiraki, Zampetakis, and Zirdelis~\cite{SotirakiZZ18} presented a total search problem motivated by \DLP. They showed that it lies in the complexity class \PPP and asked whether it is complete for the complexity class \PPP.

\subsection{Our Results}\label{sec:Our-Results}
In this work, we study formalizations of \DLP as a total search problem and prove new completeness results for the classes \PPP and \PWPP.

Our starting point is the discrete logarithm problem in ``general groups'' suggested in~\cite{SotirakiZZ18}.
Given the order $s\in\mathbb{Z}$, $s > 1$, we denote by $\GG=[s]=\{0,\ldots,s-1\}$ the canonical representation of a set with $s$ elements. 
Any efficiently computable binary operation on $\GG$ can be represented by a Boolean circuit $f\colon \B^l \times \B^l \to \B^l$ that evaluates the operation on binary strings of length $l = \lceil \log(s)\rceil$ representing the elements of $\GG$.
Specifically, the corresponding binary operation $\star$ on $\GG$ can be computed by first taking the binary representation of the elements $x,y\in\GG$, evaluating $f$ on the resulting strings, and mapping the value back to $\GG$.
Note that the binary operation $\star$ induced on $\GG$ by $f$ in this way might not satisfy the group axioms and, thus, we refer to $(\GG,\star)$ as the induced \emph{groupoid} adopting the terminology for a set with a binary operation common in universal algebra.

Assuming that $(\GG,\star)$ is a cyclic group, we might be provided with the representations of the identity element $id\in\GG$ and a generator $g\in\GG$, which, in particular, enable us to efficiently access the group elements
via an indexing function $\ind\colon [s] \to \GG$ computed as the corresponding powers of $g$ (e.g. via repeated squaring).
An instance of a general \DLP is then given by a representation $(s,f)$ inducing a groupoid $(\GG,\star)$ together with the identity element $id\in\GG$, a generator $g\in\GG$, and a target $t\in\GG$; a solution for the instance $(s,f,id,g,t)$ is either an index $x\in[s]$ such that $\ind(x)=t$ or a pair of distinct indices $x, y\in[s]$ such that $\ind(x)=\ind(y)$.
Note that the solutions corresponding to non-trivial collisions in $\ind$ ensure totality of the instance irrespective of whether the induced groupoid $(\GG,\star)$ satisfies the group axioms -- the indexing function $\I$ either has a collision or it is a bijection and must have a preimage for any $t$.

The general \DLP as defined above can clearly solve \DLP in specific groups with efficient representation such as any multiplicative group $\mathbb{Z}_p^*$ of integers modulo a prime $p$, which are common in cryptographic applications.
On the other hand, it allows for remarkably unstructured instances and the connection to \DLP is rather loose --
as we noted above, the general groupoid $(\GG,\star)$ induced by the instance might not be a group, let alone cyclic.
Therefore, we refer to this search problem as \Index (see~\Cref{def:index} in~\Cref{sec:Index-PPP} for the formal definition).

A priori, the exact computational complexity of \Index is unclear.
\cite{SotirakiZZ18} showed that it lies in the class \PPP by giving a reduction to the \PPP-complete problem \Pigeon, where one is asked to find a preimage of the $0^n$ string or a non-trivial collision for a function from $\B^n$ to $\B^n$ computed by a Boolean circuit given as an input.
No other upper or lower bound on \Index was shown in~\cite{SotirakiZZ18}.
Given that \DLP can be used to construct collision-resistant hash functions~\cite{Damgard87}, it seems natural to ask whether \Index lies also in the class \PWPP, a subclass of \PPP defined by the canonical problem \Collision, where one is asked to find a collision in a shrinking function computed by a Boolean circuit given as an input.

However, a closer look at the known constructions of collision-resistant hash functions from \DLP reveals that they crucially rely on the homomorphic properties of the function $g^x=\I(x)$.
Given that $(\GG,\star)$ induced by an arbitrary instance of \Index does not necessarily posses the structure of a cyclic group, the induced indexing function $\I$ is not guaranteed to have any homomorphic properties and it seems unlikely that \Index could be reduced to any \PWPP-complete problem such as \Collision.
In \Cref{sec:Index-PPP}, we establish that the above intuition about the lack of structure is indeed correct:

\paragraph{\Cref{thm:Index_is_PPP_complete}.}\emph{\Index is \PPP-complete.}
\smallskip

On the other hand, we show that, by introducing additional types of solutions in the \Index problem, we can enforce sufficient structure on the induced groupoid $(\GG,\star)$ that allows for a reduction to the \PWPP-complete problem \Collision.
First, we add a solution type witnessing that the coset of $t$ is not the whole $\GG$,  i.e., that $\{t\star a\mid a\in\GG\}\neq\GG$, which cannot be the case in a group.
Specifically, a solution is also any pair of distinct $x,y\in[s]$ such that $t\star\I(x)=t\star\I(y)$.
Second, we add a solution enforcing some form of homomorphism in $\I$ with respect to $t$.
Specifically, a solution is also any pair of $x,y\in[s]$ such that $\I(x)=t\star\I(y)$ and $\I(x-y \mod s)\neq t$.
The second type of solution is motivated by the classical construction of a collision-resistant hash function from \DLP by Damg\aa{}rd~\cite{Damgard87}.
Notice that if there are no solutions of the second type then any pair $x,y$ such that $\I(x)=t\star\I(y)$ gives rise to the preimage of $t$ under $\I$ by simply computing $x-y \mod s$.
We refer to the version of $\Index$ with the additional two types of solutions as \DLog (see~\Cref{def:problem_DLOG} in~\Cref{sec:DLog-PWPP} for the formal definition), as it is in our opinion closer to the standard \DLP in cyclic groups compared to the significantly less structured \Index.\footnote{To clarify our terminology, note that the problem $\mathsf{DLOG}$ from Sotiraki et al.~\cite{SotirakiZZ18} is actually a variant of our \Index.}

Since \DLog is a relaxation of \Index obtained by allowing additional types of solutions, it could be the case that we managed to reduce \DLog to \Collision simply because \DLog is trivial.
Note that this is not the case since \DLog is at least as hard as \DLP in any cyclic group with an efficient representation, where \DLP would naturally give rise to an instance of \DLog with a unique solution corresponding to the solution for the \DLP.
In \Cref{sec:DLog-PWPP}, we establish that \DLog is at least as hard as the problem of finding a non-trivial collision in a shrinking function:

\paragraph{\Cref{thm:DLog_is_PWPP_complete}.}\emph{\DLog is \PWPP-complete.}

\paragraph{Alternative characterizations of \PWPP.} Our \PWPP-completeness result for \DLog is established via a series of reductions between multiple intermediate problems, which are thus also \PWPP-complete.
We believe this characterization will prove useful in establishing further \PWPP-completeness results.
These new \PWPP-complete problems are defined in~\Cref{sec:DLog-PWPP} and an additional discussion is provided in~\Cref{sec:New-Characterizations-of-PWPP}.

\paragraph{Implications for cryptographic lower bounds for subclasses of \TFNP.}
It was shown already by Papadimitriou~\cite{Papadimitriou94} that cryptographic hardness might serve as basis for arguing the existence of average-case hardness in subclasses of \TFNP.
A recent line of work attempts to show such cryptographic lower bounds for subclasses of \TFNP under increasingly more plausible cryptographic hardness assumptions~\cite{Jerabek16,BitanskyPR15,GargPS16,HubacekY20,HubacekNY17,KomargodskiS20,ChoudhuriHKPRR19,ChoudhuriHKPRR19a,EphraimFKP20a,BitanskyG20,LombardiV20,HubacekKKS20,JawaleKKZ21}.
However, it remains an open problem whether \DLP can give rise to average-case hardness in subclasses of \TFNP other than \PWPP and \PPP.
Our results highlight that any attempt at basing average-case hardness in subclasses of \TFNP (other than \PWPP and \PPP) on the average-case hardness of the discrete logarithm problem must exploit its structural properties beyond what is necessary for constructions of collision-resistant hash functions.

\paragraph{Witnessing totality of number theoretic problems.}
In~\Cref{sec:Totality-Number-Theory}, we discuss some of the issues that arise when defining  total search problems corresponding to actual problems in computational number theory.
First, we highlight some crucial distinctions between the general \DLog as defined in~\Cref{def:problem_DLOG} and the discrete logarithm problem in multiplicative groups $\Z_p^*$.
In particular, we argue that the latter is unlikely to be \PWPP-complete.

Second, we clarify the extent to which our reductions exploit the expressiveness allowed by the representations of instances of \DLog and \Index.
In particular, both the reduction from \Collision to \DLog and from \Pigeon to \Index output instances that induce groupoids unlikely to satisfy group axioms and, therefore, do not really correspond to \DLP.
Additionally, we revisit the problem \Blichfeldt introduced in~\cite{SotirakiZZ18} and show that it also exhibits a similar phenomenon in the context of computational problems on integer lattices.

\section{Preliminaries}\label{sec:Preliminaries}
We denote by $[m]$ the set $\{0,1,\dots,m-1\}$, by $\mathbb{Z}^+$ the set $\{1,2,3,\dots\}$ of positive integers, and by $\mathbb{Z}_0^+$ the set $\{0,1,2,\dots\}$ of non-negative integers.
For two strings $u,v\in\{0,1\}^*$, $u\,||\,v$ stands for the concatenation of $u$ and $v$.
When it is clear from the context, we omit the operator $||$, e.g., we write $0x$ instead of $0\,||\,x$. 
The standard XOR function on binary strings of equal lengths is denoted by $\oplus$. 

\paragraph{Bit composition and decomposition.}
Throughout the paper, we often make use of the bit composition and bit decomposition functions between binary strings of length $k$ and the set $[2^k]$ of non-negative integers less then $2^k$.
We denote these functions $\bc^k$ and $\bd^k$. 
Concretely, $\bc^k: \{0,1\}^k \to [2^k]$ and $\bd^k: [2^k] \to \{0,1\}^k$. 
Formally, for $x=x_1x_2\dots x_k \in \{0,1\}^k$, we define $\bc^k(x)=\sum_{i=0}^{k-1} x_{k-i}2^i.$
The function $\bc^k$ is bijective and we define the function $\bd^k$ as its inverse, i.e., for $a \in [2^k]$, $\bd^k(a)$ computes the unique binary representation of $a$ with leading zeroes such that its length is $k$. When clear from the context, we omit $k$ and write simply $\bc$ and $\bd$ to improve readability. 
At places, we work with the output of $\bd^k$ without the leading zeroes.  
We denote by $\bd_0: \mathbb{Z}_0^+ \to \{0,1\}^*$ the standard function which computes the binary representation without the leading zeroes.

\paragraph{\TFNP and some of its subclasses.}
A \emph{total $\NP$ search problem} is a relation $S\subseteq\B^*\times\B^*$ such that: 1)
the decision problem whether $(x,y)\in S$ is computable in polynomial-time in $|x|+ |y|$, and 2) there exists a polynomial $q$ such that for all $x\in\B^*$, there exists a $y\in\B^*$ such that $(x,y)\in S$ and $|y| \leq q(|x|)$.
The class of all total $\NP$ search problems is denoted by $\TFNP$.
To avoid unnecessarily cumbersome phrasing, throughout the paper, we define total $\NP$ search relations implicitly by presenting the set of valid \emph{instances} $X\subseteq \B^*$ and, for each instance $i\in X$, the set of admissible \emph{solutions} $Y_i\subseteq\B^*$ for the instance $i$.
It is then implicitly assumed that, for any invalid instance $i\in\B^*\setminus X$, we define the corresponding solution set as $Y_i=\B^*$.

Let $S,T\subseteq\B^*\times\B^*$ be total search problems.
A \emph{reduction from $S$ to $T$} is a pair of polynomial-time computable functions $f, g\colon\B^*\to\B^*$ such that, for all $x,y\in\B^* $ if $(f(x),y) \in T$ then $(x,g(y))\in S$.
In case there exists a reduction from $S$ to $T$, we say that \emph{$S$ is reducible to $T$}.
The above corresponds to so-called polynomial-time many-one (or Karp) reductions among decision problems in the context of search problems.
In the rest of the paper, we consider only such reductions.

\begin{definition}[\Pigeon and \PPP~\cite{Papadimitriou94}]\label{def:problem_Pigeon}
The search problem \Pigeon is defined via the following relation of instances and solutions.
\begin{description}
    \item[Instance:] A Boolean circuit $\circuit{C}$ with $n$ inputs and $n$ outputs. 
    \item[Solution:] One of the following:
        \begin{enumerate}
            \item $u \in \{0,1\}^n$ such that $\circuit{C}(u) = 0^n$,
            \item distinct $u,v \in \{0,1\}^n$ such that $\circuit{C}(u) = \circuit{C}(v)$.
        \end{enumerate}
\end{description}
The class of all total search problems reducible to \Pigeon is called \PPP.
\end{definition}

\begin{definition}[\Collision and \PWPP~\cite{Jerabek16}]\label{def:problem_Collision}
The search problem \Collision is defined via the following relation of instances and solutions.
\begin{description}
    \item[Instance:] A Boolean circuit $\circuit{C}$ with $n$ inputs and $m$ outputs with $m < n$.
    \item[Solution:] Distinct $u,v \in \{0,1\}^n$ such that $\circuit{C}(u) = \circuit{C}(v)$.
\end{description}
The class of all total search problems reducible to \Collision is called \PWPP.
\end{definition}

\section{\DLog is \PWPP-complete}\label{sec:DLog-PWPP}

In this section, we define \DLog, a total search problem associated to \DLP and show that it is \PWPP-complete. 
Our reductions give rise to additional new \PWPP-complete problems \Dove and \Claw, which we discuss further  in~\Cref{sec:New-Characterizations-of-PWPP}.

Similarly to Sotiraki et al.~\cite{SotirakiZZ18}, we represent a binary operation on $\GG=[s]=\{0,\ldots,s-1\}$ by a Boolean circuit $f\colon \{0,1\}^{l}\times \{0,1\}^{l} \to \{0,1\}^l$, where $l = \lceil \log(s)\rceil.$
Given such a representation $(s,f)$, we define a binary operator $f_\GG: [s] \times [s] \to [2^l]$ for all $x,y\in[s]$ as
$
f_\GG(x,y) = \bc(f(\bd(x), \bd(y)))
$
using the bit composition (resp. decomposition) function $\bc$ (resp. $\bd$) defined in~\Cref{sec:Preliminaries}.
We denote by $(\GG,\star)$ the groupoid induced by $f$, where $\star\colon [s]\times[s]\to [s]$ is the binary operation closed on $[s]$ obtained by extending the operator $f_\GG$ in some fixed way, e.g., by defining $x\star y=1$ for all $x,y\in[s]$ such that $f_\GG(x,y)\not\in[s]$.

If the induced groupoid $(\GG,\star)$ was a cyclic group then we could find the indices of the identity element $id\in[s]$ and a generator $g\in[s]$.
Moreover, we could use $g$ to index the elements of the group $(\GG,\star)$, e.g., in the order of increasing powers of $g$, and the corresponding \emph{indexing function} $\ind\colon [s] \to [2^l]$ would on input $x$ return simply the $x$-th power of the generator $g$. 
We fix a canonical way of computing the $x$-th power using the standard square-and-multiply method as defined in~\Cref{algo:Ig}.
The algorithm first computes $(x_m, x_{m-1}, \dots, x_{1})=\bd_0(x)$, i.e., the binary representation of the exponent $x$ without the leading zeroes for some $m \leq l$, and it then proceeds with the square-and-multiply method using the circuit $f$.
As explained above, $f$ implements the binary group operation.
Hence, $f(r,r)$ corresponds to squaring the intermediate value $r$ and $f(g,r)$ corresponds to multiplication of the intermediate value $r$ by the generator $g$.

With the above notation in place, we can give the formal definition of \DLog.

\begin{algorithm}[t!]
\caption{Computation of the $x$-th power of the generator $g\in[s]$ of a groupoid $(\GG,\star)$ of size $s\in\N$ induced by $f\colon\{0,1\}^{2\lceil \log(s)\rceil}\to\{0,1\}^{\lceil \log(s)\rceil}$ with identity $id\in[s]$.}\label{algo:Ig}
\begin{algorithmic}[1]
\Procedure{$\I(x)$}{}
\State $(x_m,\dots, x_1) \gets \bd_0(x)$
\State $r\gets \bd(id)$
\State $g \gets \bd(g)$
\For{$i$ \textbf{from} $m$ \textbf{to} $1$}
\State $r\gets f(r,r)$
\If{$x_i = 1$}
\State $r \gets f(g,r)$
\EndIf
\EndFor
\State \textbf{return} $\bc(r)$
\EndProcedure
\end{algorithmic}
\end{algorithm}

\begin{definition}[\DLog]\label{def:problem_DLOG}
The search problem \DLog is defined via the following relation of instances and solutions.
\begin{description}
    \item[Instance:] A tuple $(s,f,id,g,t)$, where $s\in\Z^+$ is a size parameter such that $s\ge 2$ and $f$ is a Boolean circuit with $2\lceil \log(s)\rceil$ inputs and  $\lceil \log(s)\rceil$ outputs, i.e., $(s,f)$ represent a groupoid $(\GG, \star)$, and $id,g,t\in[s]$ are some indices of elements in $\GG$.
    \item[Solution:] One of the following:
    \begin{enumerate}
        \item $x \in [s]$ such that $\ind(x) = t$,
        \item $x,y \in [s]$ such that $f_\GG(x,y)\geq s$,
        \item distinct $x, y \in [s]$ such that $\ind(x) = \ind(y)$,
        \item distinct $x, y \in [s]$ such that $f_\GG(t, \I(x)) = f_\GG(t, \I(y))$,
        \item $x, y \in [s]$ such that $\I(x) = f_\GG(t, \I(y))$ and $\I(x-y \bmod s) \neq t$.
    \end{enumerate}
\end{description}
\end{definition}

The first type  of a solution in \DLog corresponds to the discrete logarithm of $t$. 
Since we cannot efficiently verify that the input instance represents a group with the purported generator $g$, additional types of solutions had to be added in order to guarantee that \DLog is total.
Note that any solution of these additional types witnesses that the instance does not induce a group, since for a valid group these types cannot happen.
Nevertheless, the first three types of solutions are sufficient to guarantee the totality of \DLog.
The last two types of solutions make \DLog to lie in the class \PWPP and are crucial for correctness of our reduction from \DLog to \Collision presented in \Cref{sec:DLog-In-PWPP}.
At the end of this section, we provide further discussion of \DLog and some of its alternative definitions.

In~\Cref{sec:DLog-is-PWPP-hard}, we show that \DLog is \PWPP-hard.
In~\Cref{sec:DLog-In-PWPP}, we show that \DLog lies in \PWPP.
Therefore, we prove \Cref{thm:DLog_is_PWPP_complete}.

\begin{theorem}\label{thm:DLog_is_PWPP_complete}
\DLog is \PWPP-complete. 
\end{theorem}

\paragraph{Alternative types of violations in \DLog.}
Since the last type of solution in \DLog implies that the associative property does not hold for the elements $t, \ind(x)$, and $\ind(y)$, one could think about changing the last type of solution to finding $x,y,z \in [s]$ such that $f_\GG(x, f_\GG(y,z)) \neq f_\GG(f_\GG(x,y),z)$ to capture violations of the associative property directly.
However, our proof of \PWPP-hardness (\Cref{sec:DLog-is-PWPP-hard}) would fail for such alternative version of \DLog and we do not see an alternative way of reducing to it from the \PWPP-complete problem \Collision.
In more detail, any reduction from \Collision to \DLog must somehow embed the instance $\circuit{C}$ of \Collision in the circuit $f$ in the constructed instance of \DLog.
However, a refutation of the associative property of the form $f(x, f(y, z)) \neq f(f(x, y), z)$ for some $x,y,$ and $z$  might simply correspond to a trivial statement $C(u) \neq C(v)$ for some $u\neq v$, which is unrelated to any non-trivial collision in $\circuit{C}$.

\paragraph{Explicit $\I$.}
A natural question about our definition of \DLog is whether its computational complexity changes if the instance additionally contains an explicit circuit computing the indexing function $\I$.
First, the indexing function $\I$ could then be independent of the group operation $f$ and, thus, the reduction from \Collision to such variant of \DLog would become trivial by defining the indexing function $\I$ directly via the \Collision instance $\circuit{C}$.
On the other hand, the core ideas of the reduction from \DLog to \Collision would remain mostly unchanged as it would have to capture also $\I$ computed by \Cref{algo:Ig}.
Nevertheless, we believe that our version of \DLog with an implicit $\I$ computed by the standard square-and-multiply algorithm strikes the right balance in terms of modeling an interesting problem.
The fact that it is more structured than the alternative with an explicit $\I$ makes it significantly less artificial and relevant to the discrete logarithm problem, which is manifested especially in the non-trivial reduction from \DLog to \Collision in \Cref{sec:DLog-In-PWPP}.

\subsection{\DLog is \PWPP-hard}\label{sec:DLog-is-PWPP-hard}

To show that \DLog is \PWPP-hard, we reduce to it from the \PWPP-complete problem \Collision (see~\Cref{def:problem_Collision}).
Given an instance $\circuit{C}\colon\B^n\to\B^{n-1}$ of \Collision, our reduction to \DLog defines a representation $(s,f)$ of a groupoid $(\GG,\star)$ and the elements $id,g$, and $t$ such that we are able to extract some useful information about \circuit{C} from any non-trivial collision $\I(x)=\I(y)$ in the indexing function $\I$ computed by~\Cref{algo:Ig}.
The main obstacle that we need to circumvent is that, even though the computation performed by $\I$ employs the circuit $f$ representing the binary operation in the groupoid, it has a very restricted form.
In particular, we need to somehow define $f$ using $\circuit{C}$ so that there are no collisions in $\ind$ unrelated to solutions of the instance of \Collision.
To sidestep some of the potential issues when handling an arbitrary instance of \Collision, we reduce to \DLog from an intermediate problem we call \Dove.


\begin{definition}[\Dove]\label{def:problem_Dove}
The search problem \Dove is defined via the following relation of instances and solutions.
\begin{description}
    \item[Instance:] A Boolean circuit $\circuit{C}$ with $n$ inputs and $n$ outputs.
    \item[Solution:] One of the following:
    \begin{enumerate}
        \item $u \in \{0,1\}^n$ such that $\circuit{C}(u) = 0^n$,
        \item $u \in \{0,1\}^n$ such that $\circuit{C}(u) = 0^{n-1}1,$
        \item distinct $u,v \in \{0,1\}^n$ such that $\circuit{C}(u) = \circuit{C}(v)$,
        \item distinct $u,v \in \{0,1\}^n$ such that $\circuit{C}(u) = C(v) \oplus 0^{n-1}1$.
    \end{enumerate}
\end{description}
\end{definition}

It is immediate that \Dove is a relaxation of \Pigeon (cf.~\Cref{def:problem_Pigeon}) with two additional new types of solutions -- the cases 2 and 4 in the above definition.
Similarly to case 1, case 2 corresponds to a preimage of a fixed element in the range.
Case 4 corresponds to a pair of strings such that their images under $\circuit{C}$ differ only on the last bit.
Permutations for which it is computationally infeasible to find inputs with evaluations differing only on a prescribed index appeared in the work of Zheng, Matsumoto, and Imai~\cite{ZhengMI90} under the term \emph{distinction-intractable} permutations.
Zheng~et~al.~showed that distinction-intractability is sufficient for collision-resistant hashing.
Note that we employ distinction-intractability in a different way than \cite{ZhengMI90}.
In particular, their construction of collision-resistant hash from distinction-intractable permutations could be leveraged towards a reduction from \Dove to \Collision (proving \Dove is contained in \PWPP)~-- we use \Dove as an intermediate problem when reducing from \Collision to \DLog (proving \PWPP-hardness of \DLog).
In the overview of the reduction from \Dove to \DLog below, we explain why distinction-intractability seems as a natural choice for our definition of \Dove.

\paragraph{Reducing \Dove to \DLog.} Let $\circuit{C}: \{0,1\}^n \to \{0,1\}^n$ be an arbitrary instance of \Dove. 
Our goal is to construct an instance $G=(s,f,id,g,t)$ of \DLog such that any solution to $G$ provides a solution to the original instance \circuit{C} of \Dove.  
The key step in the construction of $G$ is a suitable choice of the circuit $f$ since it defines both $\I$ and $f_\GG$. 
Our initial observation is that, by the definition of $\I$ (\Cref{algo:Ig}), the circuit $f$ is only applied on specific types of inputs during the computation of $\I(x)$.
Specifically:
\begin{itemize}
    \item In each loop, $f(r,r)$ is computed for some $r\in\B^*$. We denote $f$ restricted to this type of inputs by $\fo$,~\ie $\fo(r)=f(r,r)$.
    \item If the corresponding bit of $x$ is one then $f(g,r)$ is computed with fixed $g\in\B^*$ and some $r\in\B^*$.
    We denote $f$ restricted to this type of inputs by $\fl$,~\ie $\fl(r)=f(g,r)$.
\end{itemize}
Hence, using the above notation, the computation of $\I(x)$ simply corresponds to an iterated composition of the functions $\fo$ and $\fl$ depending on the binary representation of $x$ evaluated on $id$ (e.g., $\I(\bc(101))=\fl\circ \fo\circ \fo\circ \fl\circ \fo(\bd(id))$).
Exploiting the observed structure of the computation of $\I$, our approach is to define $\fo$ and $\fl$ (i.e., the corresponding part of $f$) using the circuit \circuit{C} so that we can extract some useful information about  \circuit{C} from any non-trivial collision $\I(x)=\I(y)$ (\ie from a solution to \DLog, case 3).

The straightforward option is to set $\fo(r) = \fl(r) = \circuit{C}(r)$ for all $r\in\B^n$.
Unfortunately, such an approach fails since for all distinct $u,v\in\B^n$ with Hamming weight $l$, there would be an easy to find non-trivial collision $x=\bc(u)$ and $y=\bc(v)$ of the form $\I(x) = \bc(\circuit{C}^{n+l}(id))  = \I(y)$, which might not provide any useful information about the circuit \circuit{C}.
Hence, we define $\fo$ and $\fl$ such that $\fo \neq \fl$. 

On a high level, we set $\fo(r)=\circuit{C}(r)$ and $\fl(r)=C(h(r))$ for some function $h\colon \B^n\to\B^n$ that is not the identity as in the flawed attempt above.
Then, except for some special case, a non-trivial collision $\I(x)=\I(y)$ corresponds to the identity
$
\circuit{C}(\circuit{C}(u))=\circuit{C}(h(\circuit{C}(v)))
$
for some $u,v\in\B^n$, which are not necessarily distinct.
In particular, if $\circuit{C}(u) \neq h(\circuit{C}(v))$ then the pair of strings $\circuit{C}(u), h(\circuit{C}(v))$ forms a non-trivial collision for \circuit{C}.
Otherwise, we found a pair $u,v$ such that $\circuit{C}(u)=h(\circuit{C}(v))$ that, 
for the choice $h(y) = y \oplus 0^{n-1}1$, translates into 
$
\circuit{C}(u) = \circuit{C}(v) \oplus 0^{n-1}1,
$
i.e., a pair of inputs breaking distinction-intractability of $\circuit{C}$, and corresponds to the fourth type of solution in \Dove.
Finally, the second type of solution in \Dove captures the special case when there is no pair $u,v$ such that
$\circuit{C}(\circuit{C}(u))=\circuit{C}(h(\circuit{C}(v)))$.

The formal reduction from \Dove to \DLog establishing \Cref{lemma:Dove_To_DLog} is provided below.

\begin{lemma}\label{lemma:Dove_To_DLog}
\Dove is reducible to \DLog.
\end{lemma}
\begin{proof}
Let $\circuit{C}\colon \B^n \to \B^n$ be an arbitrary instance of \Dove.
We construct the corresponding instance $G=(s,f,id,g,t)$ of \DLog.
Set $s = 2^n, g = 0, id = 1, t = 1$ and define the circuit $f: \{0,1\}^{2n} \to \{0,1\}^n$ as follows:
\[
    f(x,y) = 
\begin{cases}
\circuit{C}(x) &\mbox{if } x=y,\\
\circuit{C}(y \oplus 1) &\mbox{if } x=g \text{ and } y \neq g,\\
x \oplus y &\mbox{otherwise, }\\
\end{cases}
\]
where $x,y \in \{0,1\}^n$. Then the general group representation $(s,g,id, f)$ with the target $t$ form an instance of \DLog problem. 
We emphasize that we can access all intermediate results in the computation of $\I$ since the whole computation is performed in polynomial time in the size of the input instance \circuit{C}. 

Now we show that any solution to this \DLog instance gives a solution to the original \Dove instance \circuit{C}. Five cases can occur: 
\begin{enumerate}
    \item The solution is $x \in [s]$ such that $\ind(x) = t.$ For our \DLog instance, $t = 1$, hence $\ind(x) = 1$. From the definition of the function $\ind$, it holds that $\bd(\ind(x)) = f(r,r) = \circuit{C}(r)$ or $\bd(\ind(x)) = f(g,r) = \circuit{C}(r\oplus 0^{n-1}1)$ for some $r \in \{0,1\}^n$. Putting these equalities together, we get that
    \[
    0^{n-1}1 = \bd(1) = \bd(\ind(x)) = \circuit{C}(y), 
    \]
    where $y=r$ or $y=r\oplus 0^{n-1}1$. So this $y \in \{0,1\}^n$ is a preimage of $0^{n-1}1$ in \circuit{C}, i.e., it is a solution to the original \Dove instance \circuit{C}, case 2. 
    
    \item The solution is a pair $x,y \in [s]$ such that $f_\GG(x,y) \geq s$. 
    But since $s = 2^n$ and $f_\GG: [s]\times[s] \to [2^n]$, this case cannot happen. 
    
    \item The solution is a pair $x,y \in [s]$ such that $x \neq y$ and $\ind(x) = \ind(y)$. 
    First, we can assume that in the computation of $\ind(x)$ and of $\ind(y), r \neq \bd(g)$ for all iterations. 
    If that was the case, then for the first such occurrence it holds that
    \[
    0^n = \bd(0)=\bd(g)=r=
    \begin{cases}
    f(r',r')=\circuit{C}(r'),\\
    f(g,r')=\circuit{C}(r'\oplus 0^{n-1}1).\\
    \end{cases}
    \]
     In both cases, we found a preimage of $0^n$ in \circuit{C}, i.e., a solution to the original \Dove instance \circuit{C}, case 1. 
     
     Further, let $(x_k, \dots, x_0) = \bd_0(x)$ and $(y_l, \dots, y_0) = \bd_0(y)$ be the binary representations of $x$ and $y$, respectively, where $x_0$ and $y_0$ are the least significant bits and $k,l < n$. 
     We use the following notation: by $r_{z_i}$ we denote the value of variable $r$ in the computation of $\ind(z)$ after the loop corresponding to the bit $z_i$. Since $x \neq y$, their binary representations are distinct as well.

     Hence, there are three possible cases:
     \begin{enumerate}
         \item There is some $i$ such that $x_i \neq y_i$. Let $j$ denote the smallest such $i$. Without loss of generality, assume that $x_j = 0$ and $y_j = 1$. Hence, it holds that $r_{x_j}=\circuit{C}(a)$ and $r_{y_j} = \circuit{C}(\circuit{C}(b)\oplus 0^{n-1}1)$ for some $a,b \in \{0,1\}^n$. 
         \begin{enumerate}
             \item If $j=0$, then it holds that $\bc(r_{x_j}) = \ind(x) = \ind(y) = \bc(r_{y_j})$, which means that
         \begin{equation}\label{eq:gen_coll_to_dlog}
         \circuit{C}(a) = \circuit{C}(\circuit{C}(b)\oplus 0^{n-1}1).
         \end{equation}
         If $x=0$, then $a=\bd(id)=\bd(1)=0^{n-1}1$, so
         \[
         \circuit{C}(0^{n-1}1) = \circuit{C}(\circuit{C}(b)\oplus 0^{n-1}1).
         \]
         Now, either $0^{n-1}1 \neq \circuit{C}(b) \oplus 0^{n-1}1$, which means a collision in \circuit{C}, i.e., a solution to the original instance \circuit{C}, case 3, or $0^{n-1}1 = \circuit{C}(b) \oplus 0^{n-1}1$, which implies $0^n = \circuit{C}(b)$ and $b$ is a preimage of $0^n$ in \circuit{C}, i.e., a solution to the original instance \circuit{C}, case 1.  
         
         If $x \neq 0$, then $j < k$ and $a = \circuit{C}(c)$ for some $c \in \{0,1\}^n$. Substituting to the above~\cref{eq:gen_coll_to_dlog}, we get that 
         \[
         \circuit{C}(\circuit{C}(c)) = \circuit{C}(\circuit{C}(b)\oplus 0^{n-1}1).
         \]
         Now, either $\circuit{C}(c) \neq \circuit{C}(b) \oplus 0^{n-1}1$, which means a collision in \circuit{C}, i.e., a solution to the original instance \circuit{C}, case 3, or $\circuit{C}(c) = \circuit{C}(b) \oplus 0^{n-1}1$, so $b,c$ differ only on the last bit, i.e., they form a solution to the original instance $\circuit{C}$, case 4.  
         
         \item Now suppose that $j \neq 0$. If $r_{x_j} = r_{y_j}$, then the proof can be reduced to the previous case $j=0$. Otherwise, since $j$ is the smallest index where $x_j$ and $y_j$ differ, we know that $(x_{j-1}, \dots, x_0) = (y_{j-1}, \dots, y_0)$ and, hence, the computation of $\ind(x)$ and of $\ind(y)$ uses exactly same steps starting from $r_{x_j}, r_{y_j}.$ Since $r_{x_j} \neq r_{y_j}$ and $\ind(x) = \ind(y)$, there must be a collision after some step. Since all these steps correspond to applying the circuit $\circuit{C}$, we can find a collision in $\circuit{C}$, i.e., a solution to the original instance $\circuit{C}$, case 3. 
         \end{enumerate}
         \item It holds that $(x_l, \dots, x_0) = (y_l, \dots, y_0)$ and $k > l$. We know that $r_{x_{l+1}} = \circuit{C}(a)$ for some $a$. In the computation of $\ind$, the variable $r$ is initialized to $\bd(id)$ at the beginning. Since $(x_l, \dots, x_0) = (y_l, \dots, y_0)$, the computation of $\ind(x)$ starting from $r_{x_{l+1}}$ uses the same steps as the whole computation of $\ind(y)$, which starts from $r=\bd(id) = \bd(1) = 0^{n-1}1$. If $0^{n-1}1 = r_{x_{l+1}} = \circuit{C}(a)$, then $a$ is a preimage of $0^{n-1}1$ in \circuit{C}, i.e., a solution to the original instance \circuit{C}, case 2. If $0^{n-1}1 \neq r_{x_{l+1}}$, then there must be a collision after some step since $\ind(x) = \ind(y)$. Since all these steps correspond to applying the circuit $\circuit{C}$, we can find a collision in \circuit{C}, i.e., a solution to the original instance $\circuit{C}$, case 3.  
         \item It holds that $(x_k, \dots, x_0) = (y_k,\dots, y_0)$ and $k < l$. Then the proof proceeds exactly same as for the case $k > l$ only with the roles of $x$ and $y$ switched. 
     \end{enumerate}
     
     \item The solution is a pair $x,y \in [s]$ such that $x \neq y$ and $f_\GG(t, \I(x)) = f_\GG(t, \I(y))$. If $t = \I(x)$, then we have that
     \[
    0^{n-1}1 = \bd(1)=\bd(t)=\bd(\I(x))=
    \begin{cases}
    f(r',r')=\circuit{C}(r'),\\
    f(g,r')=\circuit{C}(r'\oplus 0^{n-1}1),\\
    \end{cases}
    \]
    for some $r'$. In both cases, we found a preimage of $0^{n-1}1$ in \circuit{C}, i.e., a solution to the original instance \circuit{C}, case 2. Similarly, if $t = \I(y)$, then we found a solution to the original instance \circuit{C}, case 2. Otherwise, since $t \neq g$, it holds that $f_\GG(t, \I(x)) = \bc(\bd(t) \oplus \bd(\I(x)))$ and that $f_\GG(t, \I(y)) = \bc(\bd(t) \oplus \bd(\I(y)))$. By combining these equalities, we obtain that 
     \[
     \bc(\bd(t) \oplus \bd(\I(x))) = \bc(\bd(t) \oplus \bd(\I(y))),
     \]
     which implies that
     \[
     \I(x) = \I(y),
     \]
     and since $x \neq y$, we proceed as in the case $3.$ above.  
     
    \item The solution is a pair $x,y \in [s]$ such that 
    \begin{equation}\label{eq:gen_coll_to_dlog_2}
\I(x) = f_\GG(t, \I(y))
\end{equation}
and $\I(x-y \bmod s) \neq t$. If $t = \I(y)$, then we can proceed as in the case 1. above. Otherwise, since $t \neq g$, we have that
\begin{equation}\label{eq:gen_coll_to_dlog_3}
f_\GG(t, \I(y)) = \bc(\bd(t) \oplus \bd(\I(y))).     
\end{equation}
By combining \cref{eq:gen_coll_to_dlog_2} and \cref{eq:gen_coll_to_dlog_3}, we get that 
    \[
    \I(x) = \bc(\bd(t) \oplus \bd(\I(y))).
    \]
    Moreover, we know that $\I(x) = \bc(\circuit{C}(r))$ for some $r \in \{0,1\}^n$ and that $\I(y) = \bc(\circuit{C}(r'))$ for some $r' \in \{0,1\}^n$. Substituting to the previous relationship and using the fact the $\bc$ and $\bd$ are bijections inverse to each other, we get that 
    \[
    \circuit{C}(r) = \bd(t) \oplus \circuit{C}(r') = \bd(1) \oplus \circuit{C}(r') = 0^{n-1}1 \oplus \circuit{C}(r').    
    \]
    Hence, the strings $r, r'$ differ only on the last bit, i.e., they form a solution to the original instance $\circuit{C}$, case 4.
    \qedhere
\end{enumerate}
\end{proof}

\paragraph{\PWPP-hardness of \Dove.} Next, we show that, by introducing additional types of solutions into the definition of \Pigeon, we do not make the corresponding search problem too easy --
\Dove is at least as hard as any problem in \PWPP.
Our reduction from \Collision to \Dove is rather syntactic and natural.
In particular, it results in instances of \Dove with all solutions being of a single type, corresponding to collisions of the original instance of \Collision.

\begin{lemma}\label{lemma:Collision_To_Dove}
\Collision is reducible to \Dove.
\end{lemma}
\begin{proof}
We start with an arbitrary instance $\circuit{C}: \{0 ,1\}^n \to \{0,1\}^m$ with $m < n$ of \problem{Collision}. Moreover, we can assume that $m = n-1$ because otherwise we can pad the output with zeroes, which preserves the collisions. We construct a circuit $\circuit{V}: \{0,1\}^{2n} \to \{0,1\}^{2n}$, considered as an instance of \Dove, as follows:
\[
\circuit{V}(x_1, \dots, x_{2n}) = (\circuit{C}(x_1,\dots, x_n), \circuit{C}(x_{n+1}, \dots, x_{2n}),1,1),
\]
where $x_i \in \{0,1\}$. The construction is valid since the new circuit \circuit{V} can be constructed in polynomial time with respect to the size of \circuit{C}. Now we show that any solution to the above instance \circuit{V} of \Dove gives a solution to the original \problem{Collision} instance \circuit{C}. Four cases can occur:
\begin{enumerate}
    \item The solution to \circuit{V} is $(x_1,\dots, x_{2n}) \in \{0,1\}^{2n}$ such that $\circuit{V}(x_1,\dots,x_{2n}) = 0^{2n}$.
    From the definition of the circuit \circuit{V}, the last bit of the output is always 1.
    Hence, this case cannot happen.
    
     \item The solution to \circuit{V} is $(x_1,\dots, x_{2n}) \in \{0,1\}^{2n}$ such that $\circuit{V}(x_1,\dots,x_{2n}) = 0^{n-1}1 = (0,0,\dots,0,1)$. From the definition of the circuit \circuit{V}, the next-to-last bit of the output is always 1 and, hence, this case cannot happen.  
    
    \item The solution to \circuit{V} is $x=(x_1, \dots, x_{2n}), y=(y_1, \dots, y_{2n}) \in \{0,1\}^{2n}$ such that $x\neq y$ and $\circuit{V}(x) = \circuit{V}(y)$. From the definition of the circuit \circuit{V}, it holds that
    \[
    \circuit{C}(x_1, \dots, x_n) = \circuit{C}(y_1, \dots, y_n)
    \]
    and 
    \[
    \circuit{C}(x_{n+1}, \dots, x_{2n}) = \circuit{C}(y_{n+1}, \dots, y_{2n}).
    \]
    Since $x \neq y$, either $(x_1, \dots, x_n) \neq (y_1, \dots, y_n)$ or $(x_{n+1}, \dots, x_{2n}) \neq (y_{n+1}, \dots, y_{2n})$. In both cases, we found a collision, i.e., a solution to the original instance \circuit{C}. 
    
    \item The solution to \circuit{V} is $x, y \in \{0,1\}^{2n}$ such that $\circuit{V}(x) = \circuit{V}(y) \oplus 0^{n-1}1$, i.e., their evaluations differ only on the last bit. From the definition of the circuit \circuit{V}, the last bit of the output is always 1 and, hence, this case cannot happen.\qedhere
\end{enumerate}
\end{proof}

The above \Cref{lemma:Dove_To_DLog} and \Cref{lemma:Collision_To_Dove} imply that \DLog is \PWPP-hard and we conclude this section with the corresponding corollary.
\begin{corollary}\label{cor:DLog_is_PWPP_hard}
\DLog is \PWPP-hard.
\end{corollary}

\subsection{\DLog Lies in \PWPP}\label{sec:DLog-In-PWPP}
In order to establish that \DLog lies in \PWPP, we build on the existing cryptographic literature on constructions of collision-resistant hash functions from the discrete logarithm problem.
Specifically, we mimic the classical approach by Damg\aa{}rd~\cite{Damgard87} to first construct a family of \emph{claw-free permutations} based on \DLP and then define a collision-resistant hash using the family of claw-free permutations.\footnote{
In principle, it might be possible to adapt any alternative known construction of collision-resistant hash from DLP; e.g., the one by Ishai, Kushilevitz, and Ostrovsky~\cite{IshaiKO05}, which goes through the intermediate object of \emph{homomorphic one-way commitments}.
However, this would necessitate not only the corresponding changes in the definition of \DLog but also an alternative proof of its \PWPP-hardness.
}
Recall that a family of claw-free permutations is an efficiently sampleable family of pairs of permutations such that given a ``random'' pair $\claw{0}$ and $\claw{1}$ of permutations from the family, it is computationally infeasible to find a \emph{claw} for the two permutations,~\ie inputs $u$ and $v$ such that $\claw{0}(u)=\claw{1}(v)$.
We formalize the corresponding total search problem, which we call \Claw, below.

\begin{definition}[\Claw]\label{def:problem_Claw}
The search problem \Claw is defined via the following relation of instances and solutions.
\begin{description}
    \item[Instance:] A pair of Boolean circuits $\claw{0}, \claw{1}$ with $n$ inputs and $n$ outputs.
    \item[Solution:] One of the following:
    \begin{itemize}
        \item $u,v \in \{0,1\}^n$ such that $\claw{0}(u) = \claw{1}(v)$,
        \item distinct $u,v \in \{0,1\}^n$ such that $\claw{0}(u) = \claw{0}(v)$,
        \item distinct $u,v \in \{0,1\}^n$ such that $\claw{1}(u) = \claw{1}(v)$.
    \end{itemize}
\end{description}
\end{definition}

The first type of solution in \Claw corresponds to finding a claw for the pair of functions $\claw{0}$ and $\claw{1}$.
As we cannot efficiently certify that both $\claw{0}$ and $\claw{1}$ are permutations, we introduce the second and third type of solutions which witness that one of the functions is not bijective.
In other words, the second and third type of solution ensure the totality of \Claw.

Similarly to~\cite{Damgard87}, our high-level approach when reducing from \DLog to \Collision is to first reduce from \DLog to \Claw and then from \Claw to \Collision.
Although, we cannot simply employ his analysis since we have no guarantee that 1) the groupoid induced by an arbitrary \DLog instance is a cyclic group and 2) that an arbitrary instance of \Claw corresponds to a pair of permutations. 
It turns out that the second issue is not crucial.
It was observed by Russell~\cite{Russell95} that the notion of claw-free \emph{pseudopermutations} is sufficient for collision-resistant hashing.
Our definition of \Claw corresponds exactly to the worst-case version of breaking claw-free pseudopermutations as defined by~\cite{Russell95}.
As for the first issue, we manage to provide a formal reduction from \DLog to \GeneralClaw, a variant of \Claw defined below.

\begin{definition}[\GeneralClaw]\label{def:problem_GeneralClaw}
The search problem \GeneralClaw is defined via the following relation of instances and solutions.
\begin{description}
    \item[Instance:] A pair of Boolean circuits $\claw{0}, \claw{1}$ with $n$ inputs and $n$ outputs and $s \in \mathbb{Z}^+$ such that $1 \leq s < 2^n$.
    \item[Solution:] One of the following:
    \begin{enumerate}
        \item $u,v \in \{0,1\}^n$ such that $\bc(u) < s$, $\bc(v) < s$, and  $\claw{0}(u) = \claw{1}(v)$,
        \item distinct $u,v \in \{0,1\}^n$ such that $\claw{0}(u) = \claw{0}(v)$,
        \item distinct $u,v \in \{0,1\}^n$ such that $\claw{1}(u) = \claw{1}(v)$,
        \item $u \in \{0,1\}^n$ such that $\bc(u) < s$ and $\bc(\claw{0}(u)) \geq s,$
        \item $u \in \{0,1\}^n$ such that $\bc(u) < s$ and $\bc(\claw{1}(u)) \geq s.$  
    \end{enumerate}
\end{description}
\end{definition}

The main issue that necessitates the introduction of additional types of solutions in the definition of \GeneralClaw (compared to \Claw) is that the possible solutions to an instance of \DLog are not from the whole domain $[2^n]$ but they must lie in $[s]$.

Below, we give the formal reduction from \DLog to \GeneralClaw followed by~\Cref{lemma:GeneralClaw_To_Collision} establishing that \GeneralClaw lies in \PWPP.

\begin{lemma}\label{lemma:DLOG_to_GeneralClaw}
\DLog is reducible to \GeneralClaw. 
\end{lemma}
\begin{proof}
We start with an arbitrary instance $G=(s,g,id,f,t)$ of \DLog. Let $n = \lceil \log(s) \rceil$. We define $\claw{0}: \{0,1\}^n \to \{0,1\}^n$ and $\claw{1}: \{0,1\}^n \to \{0,1\}^n$ as follows:
\[
\claw{0}(u) = 
\begin{cases}
\bd(\I(\bc(u))) &\mbox{if } \bc(u) < s,\\
u &\mbox{otherwise, }\\
\end{cases}
\]
and
\[
\claw{1}(u) = 
\begin{cases}
f(\bd(t), \bd(\I(\bc(u))) &\mbox{if } \bc(u) < s,\\
u &\mbox{otherwise, }\\
\end{cases}
\]
where $u \in \{0,1\}^n$. Now we show that any solution to this instance of \GeneralClaw given by $(\claw{0}, \claw{1}, s)$ gives a solution to the above instance G of \DLog. Five cases can occur: 
\begin{enumerate}
    \item The solution to $(\claw{0}, \claw{1}, s)$ is $u,v \in \{0,1\}^n$ such that $\bc(u) < s$, $\bc(v) < s$ and  $\claw{0}(u) = \claw{1}(v)$. Then, for $x = \bc(u), y = \bc(v)$, it holds that $x,y \in [s]$, so
    \[
    \claw{0}(u) = \bd(\I(\bc(u))) = \bd(\I(x)) 
    \]
    and 
    \[
    \claw{1}(v) = f(\bd(t), \bd(\I(\bc(v))) = \bd(f_\GG(t, \I(y))).
    \]
    Putting these equalities together, we get that 
    \[
    \bd(\I(x)) = \bd(f_\GG(t, \I(y))),
    \]
    and hence 
    \[
    \I(x) = f_\GG(t, \I(y)).
    \]
If $\I(x-y \bmod s) = t$, then $x-y \bmod s \in [s]$ is the discrete logarithm of $t$, i.e., a solution to the original instance G of \DLog, case 1. Otherwise, the pair $x,y$ is a solution to the original instance G of \DLog, case 5.  

    \item The solution to $(\claw{0}, \claw{1}, s)$ is $u,v \in \{0,1\}^n$ such that $u \neq v$ and $\claw{0}(u) = \claw{0}(v)$. Let $x = \bc(u)$, $y = \bc(v)$. If $x \geq s$, then from the definition of $\claw{0}$ and the fact that $x \neq y$, we get that $y < s$ and 
    \[
    u = \claw{0}(u) = \claw{0}(v) = \bd(\I(y)),
    \]
    so 
    \[
    x=\bc(u)=\I(y)
    \]
    with $x \geq s$ and $y \in [s]$. It means that after some step in the computation of $\I(y)$, it holds that $\bc(r) \geq s$. 
    We consider the first such step. Since all steps correspond to applying the circuit $f$, we have that 
    \[
    r = f(r', r'')
    \]
    for some $r', r''$ such that $\bc(r'), \bc(r'') \in [s]$. This rewrites to
    \[
    s \leq \bc(r) = f_\GG(\bc(r'), \bc(r'')).
    \]
    Hence, $\bc(r'), \bc(r'')$ is a solution to the original instance G of \DLog, case 2. We proceed analogously if $y \geq s$. Now assume that $x,y \in [s]$. Then we get that 
    \[
    \I(x) = \I(y)
    \]
    and since $x \neq y$, we found a solution to the original instance G of \DLog, case 3.
    
    \item The solution to $(\claw{0}, \claw{1}, s)$ is $u,v \in \{0,1\}^n$ such that $u \neq v$ and $\claw{1}(u) = \claw{1}(v)$. Let $x = \bc(u)$, $y = \bc(v)$. If $x \geq s$, then from the definition of $\claw{1}$ and the fact that $x \neq y$, we get that $y < s$ and 
    \[
    u = \claw{1}(u) = \claw{1}(v) = f(\bd(t),\bd(\I(y))),
    \]
    so 
    \[
    x = \bc(u) = \bc(f(\bd(t),\bd(\I(y)))) = f_\GG(t, \I(y))
    \]
    with $x \geq s$ and $y \in [s]$. If $\I(y) \geq s$, then we proceed as above in the previous case. If $\I(y) \in [s]$, then $t, \I(y)$ is a solution to the original instance G of \DLog, case 2. We proceed analogously if $y \geq s$. Now assume that $x,y \in [s]$. Then we get that 
    \[
    f(\bd(t), \bd(\I(x))) = f(\bd(t), \bd(\I(y))),  
    \]
    so
    \[
    f_\GG(t, \I(x)) = f_\GG(t, \I(y)) 
    \]
    and since $x \neq y$, we found a solution to the original instance G of \DLog, case 4. 
    
    \item The solution to $(\claw{0}, \claw{1}, s)$ is $u \in \{0,1\}^n$ such that $\bc(u) < s$ and $\bc(\claw{0}(u)) \geq s$. Let $x = \bc(u)$. Then we have that 
    \[
    s \leq \bc(\claw{0}(u)) = \I(x)
    \]
    with $x \in [s]$. Now we can proceed as in analogous situations in cases 2 and 3 above.   
    
    \item The solution to $(\claw{0}, \claw{1}, s)$ is $u \in \{0,1\}^n$ such that $\bc(u) < s$ and $\bc(\claw{1}(u)) \geq s$. Let $x = \bc(u)$. Then we have that 
    \[
    s \leq \bc(\claw{1}(u)) = \bc(f(\bd(t), \bd(\I(\bc(u)))) = f_\GG(t, \I(x)) 
    \]
    with $x \in [s]$. Now we can proceed as in the same situation in case 3 above.
    \qedhere
    \end{enumerate}
\end{proof}

Now, we give the formal reduction from \GeneralClaw to \Collision.

\begin{lemma}\label{lemma:GeneralClaw_To_Collision}
\GeneralClaw is reducible to \Collision. 
\end{lemma}
\begin{proof}
We start with an arbitrary instance $(\claw{0}, \claw{1}, s)$ of \GeneralClaw, where $\claw{0}, \claw{1}: \{0,1\}^n \to \{0,1\}^n$. We define a circuit $\circuit{C}: \{0,1\}^{n+1} \to \{0,1\}^n$ as follows:
\[
\circuit{C}(x) = \claw{x_0} \circ \claw{x_1} \circ \dots \circ \claw{x_n} (0^n),
\]
where $x = (x_0, x_1, \dots, x_n)$. The circuit \circuit{C} can be constructed in polynomial time in the size of the given instance $(\claw{0}, \claw{1}, s)$ of \GeneralClaw. 
Now we show that any solution to this instance $\circuit{C}$ of \problem{Collision} gives a solution to the original instance $(\claw{0}, \claw{1}, s)$ of \GeneralClaw. There is only one type of solution for \problem{Collision}, so assume that $\circuit{C}(x) = \circuit{C}(y)$ for $x \neq y$, where $x = (x_0, x_1, \dots, x_n)$ and $y = (y_0, y_1, \dots, y_n)$. If it holds that
\[
\bc(\claw{x_i} \circ \dots \circ \claw{x_n}(0^n)) \geq s
\]
for some $0 \leq i \leq n$, then consider the largest such $i$. We emphasize that we can check this in polynomial time. We have that
\[
\bc(\claw{x_i}\circ\dots\circ \claw{x_n}(0^n)) \geq s
\]
and 
\[
\bc(\claw{x_{i+1}}\circ\dots\circ \claw{x_n}(0^n)) < s.
\]
Then, for $u = \claw{x_{i+1}}\circ\dots\circ \claw{x_n}(0^n)$, it holds that $\bc(u) < s$ and $\bc(\claw{x_i}(u)) \geq s$. So, $u$ forms a solution to the original instance $(\claw{0}, \claw{1}, s)$ of \GeneralClaw, case 4 or 5 based on the bit $x_i$. 
We proceed analogously if
\[
\bc(\claw{y_i} \circ \dots \circ \claw{y_n}(0^n)) \geq s
\]
for some $0 \leq i \leq n$. For the rest of the proof, we can assume that 
\[
\bc(\claw{x_i} \circ \dots \circ \claw{x_n}(0^n)) < s
\]
and 
\[
\bc(\claw{y_i} \circ \dots \circ \claw{y_n}(0^n)) < s
\]
for all $0 \leq i \leq n$. 
Since $x \neq y$, there is some $i$ such that $x_i \neq y_i$. If
\[
\claw{x_i}\circ\dots\circ \claw{x_n}(0^n) = \claw{y_i}\circ\dots\circ \claw{y_n}(0^n), 
\]
then the pair $u = \claw{x_{i+1}}\circ\dots\circ \claw{x_n}(0^n)$, $v = \claw{y_{i+1}}\circ\dots\circ \claw{y_n}(0^n)$ satisfies $\bc(u) < s, \bc(v) < s$ and $\claw{x_i}(u) = \claw{y_i}(v)$ with $x_i \neq y_i$, hence the pair $u,v$ forms a solution to the original instance $(\claw{0}, \claw{1}, s)$ of \GeneralClaw, case 1. Otherwise, if
\[
\claw{x_i}\circ\dots\circ \claw{x_n}(0^n) \neq \claw{y_i}\circ\dots\circ \claw{y_n}(0^n), 
\]
then there must be some $j < i$, such that 
\[
\claw{x_j}\circ\dots\circ \claw{x_n}(0^n) = \claw{y_j}\circ\dots\circ \claw{y_n}(0^n), 
\]
and we consider the largest such $j$. Then, it holds that 
\[
\claw{x_{j+1}}\circ\dots\circ \claw{x_n}(0^n) \neq \claw{y_{j+1}}\circ\dots\circ \claw{y_n}(0^n).
\]
The pair $u = \claw{x_{j+1}}\circ\dots\circ \claw{x_n}(0^n)$, $v = \claw{y_{j+1}}\circ\dots\circ \claw{y_n}(0^n)$ satisfies $u \neq v$, $\bc(u) < s$, $\bc(v) < s$ and 
$\claw{x_j}(u) = \claw{y_j}(v)$. 
So, the pair $u,v$ forms a solution to the original instance $(\claw{0}, \claw{1}, s)$ of \GeneralClaw, case 1, 2, or 3 based on the bits $x_j$ and  $y_j$. 
 \end{proof}

The above \Cref{lemma:DLOG_to_GeneralClaw} and \Cref{lemma:GeneralClaw_To_Collision} imply that \DLog lies in \PWPP and we conclude this section with the corresponding corollary.
\begin{corollary}\label{cor:DLog_lies_in_PWPP}
\DLog lies in \PWPP.
\end{corollary}

\subsection{New Characterizations of \PWPP}\label{sec:New-Characterizations-of-PWPP}
Besides \DLog, our results in~\Cref{sec:DLog-is-PWPP-hard} and~\Cref{sec:DLog-In-PWPP} establish new \PWPP-complete problems \Dove and \Claw.
Below, we provide additional discussion of these new \PWPP-complete problems.

\subsubsection{\Dove}
The chain of reductions in~\Cref{sec:DLog-PWPP} shows, in particular, that \Dove (\Cref{def:problem_Dove}) is \PWPP-complete.
The most significant property of \Dove compared to the known \PWPP-complete problems (\Pigeon or the weak constrained SIS problem defined by Sotiraki et al.~\cite{SotirakiZZ18}) is that it is not defined in terms of an explicitly shrinking function.
Nevertheless, it is equivalent to \Collision and, thus, it inherently captures some notion of compression.
Given its different structure compared to \Collision, we were able to leverage it in our proof of \PWPP-hardness of \DLog, and it might prove useful in other attempts at proving \PWPP-hardness of additional total search problems.

We emphasize that all four types of solutions in \Dove are exploited towards our reduction from \Dove to \DLog and we are not aware of a more direct approach of reducing \Collision to \DLog that avoids \Dove as an intermediate problem.
To further illustrate the importance of the distinct types of solutions in \Dove, consider the following seemingly related \PWPP-complete problem \Prefixcollision with length-preserving instances, in which we are asked to find a collision in the first $n-1$ bits of the output.

\begin{definition}[\Prefixcollision]\label{def:problem_Prefix-Collision}
The search problem \Prefixcollision is defined via the following relation of instances and solutions.
\begin{description}
    \item[Instance:] A Boolean circuit $\circuit{C}$ with $n$ inputs and $n$ outputs.
    \item[Solution:] Distinct $u,v \in \{0,1\}^n$ such that for some $z\in\B^{n-1}$ and $b,b'\in\B$ it holds that $\circuit{C}(v) = z||b$ and $\circuit{C}(u) = z||b'$.
\end{description}
\end{definition}

It is clear that \Prefixcollision is \PWPP-hard -- any Boolean circuit $\circuit{C}\colon \B^n \to \B^{n-1}$ specifying an instance of \Collision can be transformed into an equivalent instance of \Prefixcollision simply by padding the output to length $n$ by a single zero.
Similarly, \Prefixcollision reduces to \Collision{} -- any Boolean circuit $\circuit{C}\colon \B^n \to \B^{n}$ specifying an instance of \Prefixcollision can be transformed into an instance of \Collision simply by ignoring the last bit, which preserves the collisions on the first $n-1$ bits of the output of $\circuit{C}$.
However, the problem \Prefixcollision is not sufficiently structured to allow adapting our reduction from \Dove to \DLog (\Cref{lemma:Dove_To_DLog}) and we currently do not see an immediate alternative way of reducing from \Prefixcollision to \DLog.

\subsubsection{\Claw}
Russel \cite{Russell95} showed that a weakening of claw-free permutations is sufficient for collision-resistant hashing.
Specifically, he leveraged claw-free \emph{pseudopermutations}, i.e., functions for which it is also computationally infeasible to find a witness refuting their bijectivity (in addition to the hardness of finding claws).
Our definition of \Claw ensures totality by an identical existential argument -- a pair of functions with identical domain and range either has a claw or we can efficiently witness that one of the functions is not surjective.

\Claw trivially reduces to the \PWPP-complete problem \GeneralClaw and, thus, it is contained in \PWPP.
Below, we provide also a reduction from \Collision to \Claw establishing that it is \PWPP-hard.

\begin{lemma}\label{lemma:Collision-reducible-to-Claw}
\Collision is reducible to \Claw.
\end{lemma}
\begin{proof}
We start with an arbitrary instance of \Collision given by a Boolean circuit $\circuit{C}: \{0,1\}^n \to \{0,1\}^m$ with $m < n$. Without loss of generality, we can suppose that $m=n-1$ since otherwise we can pad the output with zeroes, which preserves the collisions. We construct an instance of $\Claw$ as follows:
\[
\claw{0}(x)=\circuit{C}(x)0
\]
and
\[
\claw{1}(x)=\circuit{C}(x)1.
\]
We show that any solution to this instance $(\claw{0}, \claw{1})$ of \Claw gives a solution to the original instance $\circuit{C}$ of \Collision. Three cases can occur:
\begin{enumerate}
    \item $u,v \in \{0,1\}^n$ such that $\claw{0}(u)=\claw{1}(v)$. 
    Since the last bit of $\claw{0}(u)$ is zero and the last bit of $\claw{1}(v)$ is one, this case cannot happen.
    \item $u,v \in \{0,1\}^n$ such that $u \neq v$ and $\claw{0}(u)=\claw{0}(v)$. From the definition of $\claw{0}$, we get that $\circuit{C}(u)0=\claw{0}(u)=\claw{0}(v)=\circuit{C}(v)0$, which implies that $\circuit{C}(u)=\circuit{C}(v)$. 
    Hence, the pair $u,v$ forms a solution to the original instance $\circuit{C}$ of \Collision.
    \item $u,v \in \{0,1\}^n$ such that $u \neq v$ and $\claw{1}(u)=\claw{1}(v)$. We can proceed analogously as in the previous case to show that the pair $u,v$ forms a solution to the original instance $\circuit{C}$ of \Collision.
    \qedhere
\end{enumerate}
\end{proof}

\section{\Index is \PPP-complete}\label{sec:Index-PPP}

In this section, we study the complexity of a more restricted version of \DLog that we call \Index.
In the definition of \Index, we use the notation from~\Cref{sec:DLog-PWPP} introduced for the definition of \DLog.
In particular, the function $\I$ is the same as defined in~\Cref{algo:Ig}. 

\begin{definition}[\Index]\label{def:index}
The search problem \Index is defined via the following relation of instances and solutions
\begin{description}
    \item[Instance:] A tuple $(s,f,id,g,t)$, where $s\in\Z^+$ is a size parameter such that $s\ge 2$ and $f$ is a Boolean circuit with $2\lceil \log(s)\rceil$ inputs and  $\lceil \log(s)\rceil$ outputs, i.e., $(s,f)$ represent a groupoid $(\GG, \star)$, and $id,g,t\in[s]$ are some indices of elements in $\GG$.
    \item[Solution:] One of the following:
    \begin{enumerate}
        \item $x \in [s]$ such that $\ind(x) = t$,
        \item distinct $x,y \in [s]$ such that $f_\GG(x,y) \geq s$,
        \item distinct $x, y \in [s]$ such that  $\ind(x) = \ind(y)$.
    \end{enumerate}
\end{description}
\end{definition}

It is immediate that \DLog (\Cref{def:problem_DLOG}) is a relaxation of \Index due to the additional types of solutions.
In~\Cref{sec:Index-is-PPP-hard}, we show that \Index is \PPP-hard.
In~\Cref{sec:Index-In-PPP}, we show that \Index lies in \PPP.
Therefore, we prove \PPP-completeness of \Index.

\begin{theorem}\label{thm:Index_is_PPP_complete}
\Index is \PPP-complete.
\end{theorem}

\subsection{\Index is \PPP-hard}\label{sec:Index-is-PPP-hard}

The formal reduction from the \PPP-complete problem \Pigeon to \Index is arguably the most technical part of our work.
Given a Boolean circuit $\circuit{C}\colon \{0,1\}^n \to \{0,1\}^n$ specifying an instance of \Pigeon, our main idea is to define an instance $G=(s,f,id,g,t)$ of \Index such that the induced indexing function $\I$ carefully ``emulates'' the computation of the circuit $\circuit{C}$ --
so that any solution to $G$ provides a solution to the original instance \circuit{C} of \Pigeon.
In order to achieve this, we exploit the structure of the computation induced by $\I$ in terms of evaluations of the circuit $f$ representing the binary operation in the groupoid $(\GG,\star)$.
Specifically, the computation of $\I$ gives rise to a tree labeled by the values output by $\I$ and structured by the two special types of calls to $f$ (i.e., squaring the intermediate value or multiplying it by the generator).
Our reduction constructs $f$ inducing $\I$ with the computation corresponding to a sufficiently large such tree so that its leaves can represent all the possible inputs for the instance $\circuit{C}$ of \Pigeon and the induced indexing function $\I$ outputs the corresponding evaluation of $\circuit{C}$ at each leaf.
Moreover, for the remaining nodes in the tree, $\I$ results in a bijection to ensure there are no additional solutions of the constructed instance of \Index that would be unrelated to the original instance of \Pigeon.
Below, we provide additional details of the ideas behind the formal reduction.

Similarly to the reduction from \Dove to \DLog, the key step in our construction of $G$ is a suitable choice of the circuit $f$ since it determines the function $\I$.
Recall the notation for $\fo$ and $\fl$ introduced in the reduction from \Dove to \DLog, \ie $\fo(r)=f(r,r)$ and $\fl(r)=f(g,r)$.  
We start by describing a construction of an induced groupoid $(\GG,\star)$ independent of the instance \circuit{C} of \Pigeon but which serves as a natural step towards our reduction. 

\paragraph{Constructing bijective $\I$.}
Our initial goal in the first construction is to define $\fo $ and $\fl$ and the elements $id,g\in[s]$ such that $\I$ is the identity function,~\ie such that $\I(a)=a$ for all $a\in[s]$.
To this end, our key observation is that, for many pairs of inputs $a,b\in[s]$, the computation of $\I(b)$ includes the whole computation of $\I(a)$ as a prefix (see \Cref{algo:Ig}), e.g., for all $a,b\in[s]$ such that
\begin{itemize}
    \item[--] either $\bd_0(a)$ is a prefix of $\bd_0(b)$
    \item[--] or $\bd_0(a)=y||0$ and $\bd_0(b)=y||1$ for some $y\in\B^*$.
\end{itemize}
Specifically, if $\bd_0(a)=y||0$ then $\I(a) = \bc(\fo(\bd(\I(\bc(y)))))$,
and if $\bd_0(a)=y||1$ then $\I(a) = \bc(\fl(\bd(\I(\bc(y||0)))))$.

\begin{figure}
\begin{subfigure}{.5\textwidth}
\scalebox{0.95}{
\begin{forest}
for tree={l sep=20pt,s sep=5mm}
[\text{$\,\,\,id=0\,\,\,$}, draw, 
    [\text{$0000=0$}, draw, edge label={node[midway,left] {$\fo$}}
        [\text{$0001=1$}, draw, edge=dashed, edge label={node[midway,left] {$\fl$}}
            [\text{$0010=2$}, draw, edge label={node[midway,left] {$\fo$}}
                [\text{$0011=3\,\,$},draw,edge=dashed,edge label={node[midway,left] {$\fl\,\,\,$}}
                    [\text{$0110=6$}, draw, edge label={node[midway,left] {$\fo$}} 
                	    [\text{$0111=7$},draw,edge=dashed,edge label={node[midway,left] {$\fl\,\,$}} 
                			[\text{$1110=14$}, draw, edge label={node[midway,left] {$\fo$}}
                			    [\text{$1111=15$}, draw, edge=dashed,edge label={node[midway,left] {$\fl$}}]
                			]
                		]
                		[\text{$1100=12$}, draw, edge label={node[midway,right] {$\,\,\fo$}}
                			[\text{$1101=13$}, draw,edge=dashed, edge label={node[midway,left] {$\fl$}}]
                	    ]
                    ]    
            	]
            	[\text{$0100=4$},draw,edge label={node[midway,right] {$\,\,\fo$}}
            	    [\text{$0101=5$},draw,edge=dashed,edge label={node[midway,left] {$\fl\,\,$}}
            		    [\text{$1010=10$}, draw, edge label={node[midway,left] {$\fo$}}
            		        [\text{$1011=11$}, draw, edge=dashed,edge label={node[midway,left] {$\fl$}}]
            		    ]
            		] 
            	    [\text{$1000=8$}, draw, edge label={node[midway,right] {$\,\,\fo$}}
            		    [\text{$1001=9$}, draw, edge=dashed, edge label={node[midway,left] {$\fl$}}]
            		]
            	]
            ]
        ]
    ]
]
\end{forest}}
\caption{$\I(a)=a$}\label{fig:tree1}
\end{subfigure}
\begin{subfigure}{.5\textwidth}
\scalebox{0.95}{
\begin{forest}
for tree={l sep=20pt,s sep=5mm}
[\text{$\,\,\,id=4\,\,\,$}, draw, 
    [\text{$\,\,\,4\,\,\,$}, draw, edge label={node[midway,left] {$\fo$}}
        [\text{$\,\,\,5\,\,\,$}, draw, edge=dashed, edge label={node[midway,left] 
        {$\fl$}}
            [\text{$\,\,\,6\,\,\,$}, draw, edge label={node[midway,left] {$\fo$}}
                [\text{$\,\,\,7\,\,\,$},draw,edge=dashed,edge label={node[midway,left] {$\fl\,\,\,$}}
                    [\text{$\,\,\,10\,\,\,$}, draw, edge label={node[midway,left] {$\fo$}} 
                	    [\text{$\,\,\,11\,\,\,$},draw,edge=dashed,edge label={node[midway,left] {$\fl\,\,$}} 
                			[\text{$\,\,\,15\,\,\,$}, dashed, draw, edge label={node[midway,left] {$\fo$}}
                			    [\text{$\bc(\circuit{C}(11))$}, dotted, draw, edge=dashed,edge label={node[midway,left] {$\fl$}}]
                			]
                		]
                		[\text{$\,\,\,14\,\,\,$},dashed, draw, edge label={node[midway,right] {$\,\,\fo$}}
                			[\text{$\bc(\circuit{C}(10))$}, dotted, draw,edge=dashed, edge label={node[midway,left] {$\fl$}}]
                	    ]
                    ]    
            	]
            	[\text{$\,\,\,8\,\,\,$},draw,edge label={node[midway,right] {$\,\,\fo$}}
            	    [\text{$\,\,\,9\,\,\,$},draw,edge=dashed,edge label={node[midway,left] {$\fl\,\,$}}
            		    [\text{$\,\,\,13\,\,\,$}, dashed, draw, edge label={node[midway,left] {$\fo$}}
            		        [\text{$\bc(\circuit{C}(01))$}, dotted, draw, edge=dashed,edge label={node[midway,left] {$\fl$}}]
            		    ]
            		] 
            	    [\text{$\,\,\,12\,\,\,$}, dashed, draw, edge label={node[midway,right] {$\,\,\fo$}}
            		    [\text{$\bc(\circuit{C}(00))$}, dotted, draw, edge=dashed, edge label={node[midway,left] {$\fl$}}]
            		]
            	]
            ]
        ]
    ]
]
\end{forest}}
\caption{\circuit{C} incorporated}\label{fig:tree2}
\end{subfigure}
\caption{Trees induced by the computation of $\I$.}\label{fig:trees}
\end{figure}
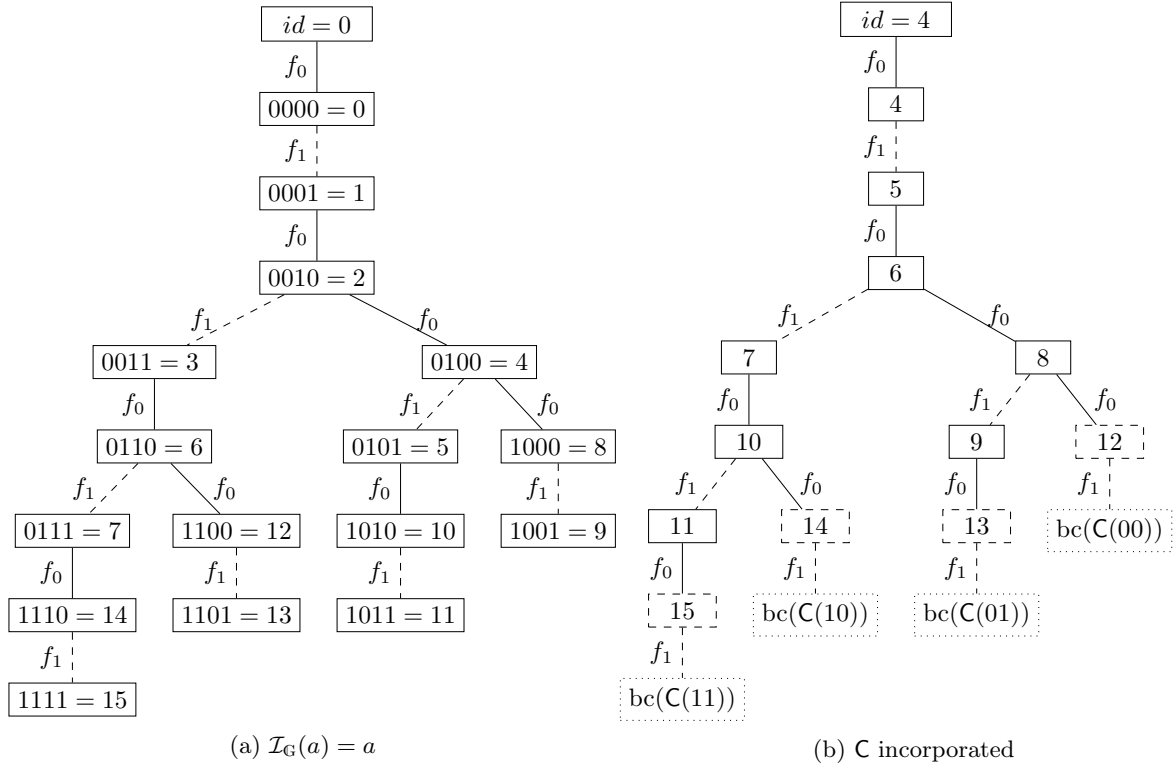

Thus, we can capture the whole computation of $\I$ on all the possible inputs from $\GG$ via a tree representing the successive calls to $f_0$ and $f_1$ based on the bit decomposition $\bd_0(a)$ of the input $a$ without the leading zeroes.
In~\Cref{fig:tree1}, we give a tree induced by the computation of $\I$ in a groupoid of order $s=16$ with $id=0$.
Solid lines correspond to the application of $\fo$ and dotted lines to application of $\fl$.
Except for the root labeled by the identity element $id$, each node of the tree corresponds to the point at which $\I$ terminates on the corresponding input $a\in[s]$, where the second value in the label of the node is the input $a$ and the first value is $\bd(a)$, \ie the binary representation of $a$ with the leading zeroes.

Note that \Cref{fig:tree1} actually suggests which functions $\fo$ and $\fl$ induce $\I$ such that $\I(a)=a$ for all $a\in[s]$.
In particular, \Cref{algo:Ig} initializes the computation of $\I$ with $r=\bd(id)=\bd(0)=0^n$ and, thus, the desired traversal of the computation tree is achieved for all inputs $a\in[s]$ by 1) $\fo$ that performs a cyclic shift of the input $r$ to the left and 2) $\fl$ that flips the last bit of the input $r$.

Similarly, the above observation allows to construct $\fo'$ and $\fl'$ such that for all $a\in[s]$ that $\I(a)=a+b \mod s$ for some fixed $b\in[s]$, which can be performed simply by setting $id=b$ and consistently ``shifting'' the intermediate value $r$ by the bit decomposition of the fixed value $b$ before and after application of the above functions $\fo$ and $\fl$. 

\paragraph{Incorporating the \Pigeon instance.}
The issue which makes it nontrivial to reduce from \Pigeon to \Index is that the functions $\fo$ and $\fl$ inducing the groupoid $(\GG,\star)$ are oblivious to the actual progress of the computation performed by $\I$.
The above discussion shows that we have some level of control over the computation of $\I$.
However, it is a priori unclear how to meaningfully incorporate the \Pigeon instance $\circuit{C}$ into the above construction achieving that $\I(a)=a$ for all $a\in[s]$.
For example, we cannot simply allow $\fo$ or $\fl$ to output $\circuit{C}(r)$ while at some internal node in the computation tree of $\I$ as this would completely break the global structure of $\I$ on the node and all its children and, in particular, could induce collisions in $\I$ unrelated to the collisions in $\circuit{C}$.
However, we can postpone the application of $\circuit{C}$ to the leaves of the tree since, for all inputs $a$ corresponding to a leaf in the tree, the computation of $\I(a)$ is not a part of the computation for $\I(b)$ for another input $b$.

Given that we are restricted to the leaves of the computation tree when embedding the computation of $\circuit{C}$ into $\I$, we must work with a big enough tree in order to have as many leaves as the $2^n$ possible inputs of the circuit $\circuit{C}\colon\B^n\to\B^n$.
In other words, the instance of \Index must correspond to a groupoid of order $s$ strictly larger than $n$.
Note that for $s=2^k$, the leaves of the tree correspond exactly to the inputs for $\I$ from the set
\[
A_o=\{a\in[2^k]\mid\exists y\in\B^{k-2}\colon \bd(a)=1||y||1\},
\]
i.e., the set of \emph{odd} integers between $2^{k-1}$ and $2^{k}$, which has size $2^{k-2}$.
Thus, in our construction, we set $s=2^{n+2}$ to ensure that there are $2^n$ leaves that can represent the domain of $\circuit{C}$.

Our goal is to define $\I$ so that its restriction to the internal nodes of the tree (non-leaves) is a bijection between $[2^{n+2}] \setminus A_o$ and $[2^{n+2}] \setminus [2^n]$.
In other words, when evaluated on any internal node of the tree, $\I$ avoids the values in $[2^n]$ corresponding to bit composition of the elements in the range of $\circuit{C}$.
If we manage to induce such $\I$ then there are no non-trivial collisions in $\ind$ involving the internal nodes -- the restrictions of $\ind$ to $A_o$ and to its complement $[2^{n+2}] \setminus A_o$ would have disjoint images and, by the bijective property of the restriction to the internal nodes of the tree, any collision in $\I$ would be induced by a collision in $\circuit{C}$.
Our construction achieves this goal by starting with $\fo$ and $\fl$ inducing $\I$ such that, for all $a\in[2^{n+2}]$, it holds that $\I(a)=a+2^n \mod 2^{n+2}$, which we already explained above.

Note that the image of the restriction of $\ind$ to the set
\[
A_e=\{a\in[2^{n+2}]\mid\exists y\in\B^{n}\colon \bd(a)=1||y||0\},
\]
i.e., the set of \emph{even} integers between $2^{n+1}$ and $2^{n+2}$, has non-empty intersection with integers in $[2^n]$ corresponding to the range of $\circuit{C}$.
Nevertheless, it is possible to locally alter the behaviour of $\fo$ and$\fl$ on $A_e$ so that $\ind$ does not map to $[2^n]$ when evaluated on $A_e$.
Then, we adjust the definition of $\fo$ and $\fl$ such that for all inputs $a\in A_o$ corresponding to a leaf of the tree, $\I(a)=\bc(C(h(a)))$ for some bijection $h$ between $A_o$ and $\B^n$ (e.g., one specific choice is simply the function that drops the first and the last bit from the binary decomposition $\bd(a)$ of $a$).
Finally, we set the target in the resulting instance of \Index to $t=0$ to ensure that the preimage of $t$ under $\I$ corresponds exactly to a preimage of $0^n$ under $\circuit{C}$.

In \Cref{fig:tree2}, we illustrate the computation tree of $\I$ corresponding to an instance of \Index produced by our reduction on input $\circuit{C}\colon \B^n\to\B^n$ for $n=2$.
Accordingly, $\GG$ is of size $s=2^{n+2}=16$ and its elements are represented by the nodes of the tree.
When compared with the tree in \Cref{fig:tree1}, the label of each node in \Cref{fig:tree2} equals the value $\I(a)$, where $a$ is the second value in the label of the node at the same position in the tree in \Cref{fig:tree1}.
Nodes belonging to $[2^s]\setminus A_e \cup A_o, A_e,$ and $A_o$ are highlighted by differing styles of edges.
Specifically, the labels of nodes with a solid edge correspond to the evaluations of the inputs from $[2^{n+1}]=[8]$, the labels of nodes with a dashed edge correspond to evaluations of the inputs from $A_e$, and the labels of nodes with a dotted edge correspond to the evaluations of the inputs from $A_o$.
Since the image of $\bc\circ\circuit{C}$ is $[2^n]=[4]$, it is straightforward to verify that any collision in  $\I$ depicted in \cref{fig:tree2} must correspond to a collision in $\circuit{C}$ and that any preimage of $t=0$ under $\I$ corresponds directly to a preimage of $0^n$ under $\circuit{C}$.

The formal reduction establishing \Cref{lemma:Pigeon_To_Index} is given below.

\begin{lemma}\label{lemma:Pigeon_To_Index}
\Pigeon is reducible to \Index.
\end{lemma}

Before proving \Cref{lemma:Pigeon_To_Index}, we describe a concrete construction of a representation of a groupoid independent of any instance of \Pigeon, but which we leverage in the proof of \Cref{lemma:Pigeon_To_Index}.

In this section, we denote $\I': \{0,1\}^n \to \{0,1\}^n$ the function
\begin{equation*}
    \I'(u) = \bd(\I(\bc(u))),
\end{equation*}
where $\I$ is induced by a representation $(s,f)$ of a groupoid $(\GG,\star)$ and elements $id,g\in[s]$. 
Recall the terminology from \Cref{sec:DLog-is-PWPP-hard},~\ie that $\fo(r)=f(r,r)$ and $\fl(r) = f(g,r)$ for all $r \in \{0,1\}^n$, where $n=\lceil \log(s) \rceil$.  
Additionally, recall that $\I$ and $\I'$ are fully determined by $\fo$ and $\fl$ as discussed in~\Cref{sec:DLog-is-PWPP-hard}.

For the constructed representation of groupoid $(\GG,\star)$, we set $s=2^n$ and $id=0$.
First, we show how to define $\fo$ and $\fl$ such that for all $v \in \{0,1\}^n,$
\begin{equation}\label{eq:pigeon_to_index_I_as_identity}
\I'(v)=v,
\end{equation}
which is equivalent to $\I(a)=a$ for all $a \in [2^n]$. 
Moreover, as it will be clear from the construction, it is enough to define $\fo$ and $\fl$ only on some subset of its potential inputs.
For all $r \in \{0,1\}^{n-1}$, we define
\begin{equation*}
\fo(0r)=r0 \quad \text{and}\quad
\fl(r0) = r1.
\end{equation*}
In other words, $\fo$ shifts the input string by one position to the left and $\fl$ changes the last bit of the input from $0$ to $1$.
Equivalently, if we interpret the functions $\fo$ and $\fl$ as functions on the corresponding integers, then $\fo$ would represent multiplying the input by two and $\fl$ would represent adding one to the input. 
We show that it is enough to define $\fo$ and $\fl$ only on the inputs of the above special form to determine the whole functions $\I$ and $\I'$.

We prove that \Cref{eq:pigeon_to_index_I_as_identity} holds for all $v \in \{0,1\}^n$ by induction on the length of $\bd_0(\bc(v))$, \ie on the length of $v$ without leading zeroes. 
For $v=0^n$, it holds that 
\[
\I'(v)=\fo(v)=\fo(0^n)=0^n=v,
\] 
so \Cref{eq:pigeon_to_index_I_as_identity} holds.
We show the inductive step first for all $v$ of the form $v=v'0$ and then for all $v$ of the form $v=v'1$.
For all $v=v'0$, we have that
\[
\I'(v)=\fo(\I'(0v'))=\fo(0v')=v'0=v, 
\]
where the first equality is from the definition of $\I'$, the second one is from the inductive hypothesis, and the third one is from the definition of $\fo$. 
Hence, \Cref{eq:pigeon_to_index_I_as_identity} holds.
Similarly, for all $v=v'1$, we have that \[
\I'(v)=\fl(\I'(v'0))=\fl(v'0)=v'1=v, 
\]
where the first equality comes from the definition of $\I'$, the second one from the inductive hypothesis and the third one from the definition of $\fl$. 
Thus, for all $v \in \{0,1\}^n,$ \Cref{eq:pigeon_to_index_I_as_identity} holds.

\Cref{fig:tree1} illustrates the tree corresponding to the computation of $\I$ induced by the above construction of (s,f,id,g) for $s=[16]$. 
Solid lines correspond to applications of $\fo$ and dotted lines to applications of $\fl$. 
For each node, the second value is the input $a$ to $\I$, which in this case equals also the output $\I(a)$, and the first value is $\bd(a)$, \ie the binary representation of $a$ with the leading zeroes.

Now, we show how to adjust the above construction to define $f'_0$ and $f'_1$ such that for a given fixed $w \in \{0,1\}^n$ and for all $v \in \{0,1\}^n$
\begin{equation}\label{eq:pigeon_to_index_I_shifted}
    \I'(v) = v + w,
\end{equation}
where $\I'$ is now determined by $f'_0$ and $f'_1$, and 
by $v + w$ we denote $\bd(\bc(v)+\bc(w) \bmod 2^n)$, \ie the standard addition with the potential carry being ignored.
Observe that \Cref{eq:pigeon_to_index_I_shifted} is equivalent to $\I(a)=a+b \bmod 2^n$ for some fixed $b = \bc(w) \in [2^n]$ and all $a \in [2^n]$. 
We implement this property by shifting the whole computation by $w$. 
To do so, we first change the identity element to $id=\bc(w)$
In the computation of $f'_0$ and $f'_1$, we first subtract $w$ and then apply the original $\fo$ or $\fl$ to the result and finally shift it back by adding $w$.

Formally, we define for all $r\in\B^n$,
\begin{equation*}
f'_0(r)=\fo(r-w)+w \quad\text{and}\quad f'_1(r) = \fl(r-w)+w,     
\end{equation*}
where $r-w$ is defined in the same manner as the addition, \ie $r-w=\bd(\bc(r)-\bc(w) \bmod 2^n)$.
We show that for all $r\in\B^n$,  \Cref{eq:pigeon_to_index_I_shifted} holds by induction on the length of $\bd_0(\bc(v))$ similarly as for the \Cref{eq:pigeon_to_index_I_as_identity}.
For $v=0^n$, we have that $\I'(v)=\I'(0^n)=f'_0(\bd(id)))=f'_0(w)=\fo(w-w)+w=\fo(0^n)+w=0^n+w=w$,
so \Cref{eq:pigeon_to_index_I_shifted} holds.
We show the inductive step first for all $v$ such that $v=v'0$ and then for all $v$ such that $v=v'1$. 
For $v=v'0$, we have that
\begin{align*}
    \I'(v)&=f'_0(\I'(0v'))=f'_0(0v'+w)=\fo(0v'+w-w)+w\\
    &=\fo(0v')+w=v'0=v+w,
\end{align*}
where the first equality comes from the definition of $\I'$, the second one from the inductive hypothesis, the third one from the definition of $f'_0$, and the fifth one from the definition of $\fo$.
Hence, \Cref{eq:pigeon_to_index_I_shifted} holds.
Similarly, for $v=v'1$, we have that 
\begin{align*}
    \I'(v)&=f'_1(\I'(v'0))=f'_1(v'0+w)=\fl(v'0+w-w)+w\\
    &=\fl(v'0)+w=v'1=v+w
\end{align*}
for analogous reasons as before. This concludes the proof that \Cref{eq:pigeon_to_index_I_shifted} holds for all $r\in\B^n$. 

We can now proceed to utilize the above construction in the proof of \Cref{lemma:Pigeon_To_Index}.

\begin{proof}[Proof of \Cref{lemma:Pigeon_To_Index}]
Let $\circuit{C}: \{0,1\}^n \to \{0,1\}^n$ be an arbitrary instance  of \Pigeon. 
We construct an instance $G=(s,f,id,g,t)$ of \Index such that any solution to $G$ gives a solution to the original instance \circuit{C} of \Pigeon. 
In the rest of the proof, we denote by $\mathbb{Z}_{even}$ the subset of $\mathbb{Z}^+$ consisting of even integers and, analogously, by $\mathbb{Z}_{odd}$ we denote the subset of odd integers. 
We set $s=2^{n+2}, g=2^{n+2}-1, id=0$ and $t=0$. 
The idea is to define $f$ such that 
\begin{equation}\label{eq:Index_is_PPP_hard_new}
\I(a)=\begin{cases}
a+2^n &\mbox{if } a \in [2^{n+1}],\\
2^{n+1}+\tfrac{a}{2}&\mbox{if } a \in [2^{n+1},\dots, 2^{n+2}-1] \cap \mathbb{Z}_{even} =: A_e,\\
\bc(\circuit{C}(\bd^{n}(\frac{a-1}{2}-2^n)))&\mbox{if } a \in [2^{n+1},\dots, 2^{n+2}-1] \cap \mathbb{Z}_{odd} =: A_o\\
\end{cases}
\end{equation}

For the case $n=2$, we illustrate the structure of computation corresponding to $\I$ satisfying \Cref{eq:Index_is_PPP_hard_new} in \Cref{fig:tree2}. 
The nodes with a solid edge correspond to the set $[2^{n+1}]=[8]$, the nodes with a dashed edge correspond to $A_e$, and the nodes with a dotted edge correspond to $A_o$.
The label of each node equals the value $\I(a)$, where $a$ is the second value of the node at the same position in the tree in \Cref{fig:tree1}.

Suppose that we can define the circuit $f$ such that the induced indexing function $\I$ satisfies \Cref{eq:Index_is_PPP_hard_new}. 
Then $[s]=[2^{n+2}]=A_o \, \dot{\cup} \, ([2^{n+1}]\,\dot{\cup}\, A_e)$, where $\dot{\cup}$ denotes the disjoint union operation.  
By \Cref{eq:Index_is_PPP_hard_new}, we get that $\I$ restricted to $[2^{n+1}] \,\dot{\cup}\, A_e$ is a bijection between $[2^{n+1}] \,\dot{\cup}\, A_e$ and $[2^n, \dots, 2^{n+2}-1]$. 
Moreover, $\I$ restricted to $A_o$ outputs only values in $[2^n]$.
Hence, any collision in $\I$ or any preimage of $t=0$ under $\I$ can happen only for some values in $A_o$.
Furthermore, for values from $A_o$, the function $\I$ is defined using the input instance \circuit{C} of \Pigeon, so any collision in $\I$ or any preimage of $t$ under $\I$ give us a solution to the original instance $\circuit{C}$ of \Pigeon.
Formally, the oracle solving the above instance of \Index returns one of the following:
\begin{enumerate}

    \item $u \in [s]$ such that $\I(u)=t=0$. From the above discussion, we know it must be the case that $u \in A_o$. Hence, from the definition of $\I$, we get that 
    \begin{align*}
    0^n 
    &= \bd(0)=\bd(t)=\bd(\I(u))\\
    &=\bd(\bc(\circuit{C}(\bd^n(\tfrac{u-1}{2}-2^n))))=\circuit{C}(\bd^n(\tfrac{u-1}{2}-2^n)),
    \end{align*}
    and $\bd^n(\tfrac{u-1}{2}-2^n)$ is a solution to the original instance \circuit{C} of \problem{Pigeon}, case 1.
     \item $u,v \in [s]$ such that $u \neq v$ and $f_\GG(u,v) \geq s.$ Since $s=2^{n+2}$, this case cannot happen.
    \item $u, v \in [s]$ such that $u \neq v$ and $\I(u) = \I(v)$.
    Similarly as for the previous case, it must hold that $u, v \in A_o$. 
    Hence, from the definition of $\I$, we get that 
    \[
    \circuit{C}(\bd^n(\tfrac{u-1}{2}-2^n)) = \circuit{C}(\bd^n(\tfrac{v-1}{2}-2^n)).
    \]
    From the fact that $u \neq v $ and from the definition of the set $A_o$, it follows that $\bd^n(\tfrac{u-1}{2}-2^n) \neq \bd^n(\tfrac{v-1}{2}-2^n)$, hence the pair $\bd^n(\tfrac{u-1}{2}-2^n), \bd^n(\tfrac{v-1}{2}-2^n)$ forms a collision for $\circuit{C}$, \ie a solution to the original instance $\circuit{C}$ of \problem{Pigeon}, case 2.   
\end{enumerate}
It remains to define the circuit $f$ such that the induced indexing function $\I$ satisfies \Cref{eq:Index_is_PPP_hard_new}. 
Her, we make use of the construction defined and analysed above, where we set $w=\bd(2^n)$, \ie for which it holds that $\I(a)=a+2^n \bmod 2^{n+2}$ and that $\I'(v)=v+w$. 
It remains to adjust the definition of $f'_0$ and $f'_1$ such that we get the desired output for the values from $A_e$ and $A_o$.
To this end, we set
\begin{equation}\label{eq:Index_is_PPP_hard_def_of_f}
f(u,v) = 
\begin{cases}
11v' &\mbox{if } u=v\neq g \text{ and } v-w=01v',\\
f'_0(v) &\mbox{if } u=v \neq g \text{ and } v-w\neq 01v',\\
\circuit{C}(\bd^n(\bc(v)-2^n-2^{n+1})) &\mbox{if } u=g \text{ and } \bc(v) \in [2^{n+1}, 2^{n+2}-1],\\
f'_1(v) &\mbox{if } u=g \text{ and } v-w\neq1v'0.\\
\end{cases}
\end{equation}
Note that $f$ can be defined on the remaining inputs arbitrarily since they are not used in the computation of $\I$.

Note that for $a\in[2^{n+1}]$, only the cases 1 and 3 from the definition of $f$ in \Cref{eq:Index_is_PPP_hard_def_of_f} are used in the computation of $\I(a)$. 
Since cases 1 and 3 coincide with the previous construction, we have that $\I(a)=a+2^n \bmod 2^{n+2} = a+2^n$ for all $a \in [2^{n+1}]$, which corresponds to the first case \Cref{eq:Index_is_PPP_hard_new}.

For $a \in A_e$, it holds that $\bd(a)$ is of the form $\bd(a)=1v'0$. Hence, we get that 
\[
\I(a)=\bc(f(\I'(01v'), \I'(01v')))=\bc(f(01v'+w, 01v'+w))=\bc(11v'), 
\]
where the first equality comes from the definition of $\I$, the second one from the previous construction and the last one from the definition of $f$.
Furthermore, we have that
$\bc(11v')=2^{n+1}+2^n+\bc(v') = 2^{n+1}+\bc(01v')=2^{n+1}+\tfrac{a}{2}$, which proves the second case in \Cref{eq:Index_is_PPP_hard_new}. 

For $a \in A_o$, it holds that $\bd(a)$ is of the form $\bd(a)=1v'1$. Hence, we get that
\begin{equation}\label{eq:Index_is_PPP_Hard_001}
\I(a)=\bc(f(g, \bd(\I(\bc(1v'0))))).
\end{equation}
Moreover, it holds that $\bc(\bd(\I(\bc(1v'0))))= \I(\bc(1v'0)) \in [2^{n+1}, 2^{n+2}-1]$ from the already proved second part of the relationship in \cref{eq:Index_is_PPP_hard_new} since $\bc(1v'0) \in A_e$.
Hence, the third case from the definition of $f$ in \Cref{eq:Index_is_PPP_hard_def_of_f}  applies to \Cref{eq:Index_is_PPP_Hard_001} and we get that 
\begin{align*}
\I(a)&=\bc(\circuit{C}(\bd^n(\bc(\bd(\I(\bc(1v'0))))-2^n-2^{n+1})))\\&=\bc(\circuit{C}(\bd^n(\I(\bc(1v'0))-2^n-2^{n+1}))) \\
&= \bc(\circuit{C}(\bd^n(\tfrac{\bc(1v'0)}{2}+2^{n+1} - 2^n - 2^{n+1}))) \\
&= \bc(\circuit{C}(\bd^n(\tfrac{\bc(1v'0)}{2}-2^n))) \\
&= \bc(\circuit{C}(\bd^n(\tfrac{a-1}{2}-2^n))),
\end{align*}
which proves the last case in \cref{eq:Index_is_PPP_hard_new} and concludes the proof of~\Cref{lemma:Pigeon_To_Index}.
\end{proof}

\subsection{\Index Lies in \PPP}\label{sec:Index-In-PPP}

The main idea of our reduction from \Index to \Pigeon is analogous to the reduction in~\cite{SotirakiZZ18} from their discrete logarithm problem in ``general groups'' to \Pigeon.
Although, we need to handle the additional second type of solution for \Index, which  corresponds to $f_\GG$ outputting an element outside $\GG$.
The formal reduction proving \Cref{lemma:Index_To_Pigeon} is given below.

\begin{lemma}\label{lemma:Index_To_Pigeon}
\Index is reducible to \Pigeon.
\end{lemma}
\begin{proof}
Let $(s,f,id,g,t)$ be an arbitrary instance of \Index. Then we know that $\ind: [s] \to [2^l]$, where $l = \lceil \log(s) \rceil$. We construct a circuit $\circuit{C}: \{0,1\}^l \to \{0,1\}^l$ as follows:
\[
    \circuit{C}(x) = 
\begin{cases}
\bd(\ind(\bc(x)) - t \bmod s) &\mbox{if } \bc(x) < s,\\
x &\mbox{otherwise}.\\
\end{cases}
\]
We show that any solution to the above \problem{Pigeon} instance \circuit{C} gives a solution to the original \Index instance. There are two possible cases:
\begin{enumerate}
    \item The solution to \circuit{C} is $x \in \{0,1\}^l$ such that $\circuit{C}(x) = 0^l$. Then, from the definition of the circuit \circuit{C}, it holds that $\circuit{C}(x) = \bd(\ind(\bc(x)) - t \bmod s) = 0^l$ and $\bc(x) < s$. 
    Because the function bd is bijective and $\bd(0) = 0^l$, it must hold that $\ind(\bc(x)) - t \bmod s = 0$. 
    If $\ind(\bc(x)) \geq s$, then, from the definition of the function $\I$, we can find $u,v \in [s]$ such that $u \neq v$ and $f_\GG(u,v) \geq s$, i.e, a solution to the original \Index instance, case $2.$ Otherwise, $\ind(\bc(x)) < s$, and because $t \in [s]$, i.e., $t < s$, it must be that $\ind(\bc(x)) = t$. Hence, $\bc(x)$ is a solution to the original \Index instance, case $1.$
    
    \item The solution to \circuit{C} is a pair $x, y \in \{0,1\}^l$, such that $x \neq y$ and $\circuit{C}(x) = \circuit{C}(y)$. Then, from the definition of the circuit \circuit{C} and the bijective property of $\bd$, it must hold that $\bc(x) < s, \bc(y) < s$ and $$\ind(\bc(x)) - t = \ind(\bc(y)) - t \pmod s,$$
    which implies that
    $$
    \ind(\bc(x)) = \ind(\bc(y)) \pmod s.
    $$
    If $\ind(\bc(x)) \geq s$, then we can find a solution to the original \Index instance, case $2$, analogously as above. 
    The same holds if $\ind(\bc(y)) \geq s$.
    Otherwise, both  $\ind(\bc(x)), \ind(\bc(y)) \in [s]$, hence
    $$
    \ind(\bc(x)) = \ind(\bc(y)).
    $$
    Moreover, from $x \neq y$ and the bijectivite property of $\bc$, we get that $\bc(x) \neq \bc(y)$. Hence, the pair $\bc(x), \bc(y)$ is a solution to the original \Index instance, case $3.$
    \qedhere
\end{enumerate}
\end{proof}

\section{Ensuring the Totality of Search Problems in Number Theory}\label{sec:Totality-Number-Theory}

In this section, we discuss some of the issues that arise when defining total search problems corresponding to specific problems in computational number theory which can be solved efficiently when given access to a \PPP or \PWPP oracle.

\subsection{\DLP in Specific Groups}

In this section, we present a formalization of the discrete logarithm problem in $\Z_p^*$, i.e., the multiplicative group of integers modulo a prime $p$.
Our goal is to highlight the distinction between the general \DLog as defined in~\Cref{def:problem_DLOG} and the discrete logarithm problem in any specific group $\Z_p^*$.
In particular, we argue that the latter is unlikely to be \PWPP-complete.

\begin{definition}[\DLogp]\label{def:problem-DLogp}
The search problem \DLogp is defined via the following relation of instances and solutions.
\begin{description}
    \item[Instance:] Distinct primes $p,p_1,\ldots,p_n\in\N$, natural numbers $k_1\ldots,k_n\in\N$, and $g,y\in\Z_p^*$ such that
    \begin{enumerate}
        \item $p-1=\prod_{i=1}^{n}p_i^{k_i}$ and
        \item $g^{(p-1)/p_i}\neq 1$ for all $i\in\{1,\ldots,n\}$.
    \end{enumerate}
    \item[Solution:] An $x\in\{0,\ldots p-2\}$ such that $g^x=y$.
\end{description}
\end{definition}
The second condition in the previous definition ensures that $g$ is a generator of $\mathbb{Z}_p^*$ due to the Lagrange theorem.
If $g$ was not a generator, then there would be a $q \in \{1,\dots,p-2\}$ such that $g^q=1$. 
We consider the smallest such $q$. 
Then $(g,g^2,g^3,\dots,g^q)$ forms a subgroup of $\mathbb{Z}_p^*$, and, from the Lagrange theorem, we get that $q=p_1^{l_1} \cdot \ldots \cdot p_n^{l_n}$ such that at least one $l_i < k_i$. 
Then $g^{(p-1)/p_i}$ is a power of $g^q$, implying that $g^{(p-1)/p_i}=1$, which would be a contradiction with the second condition in~\Cref{def:problem-DLogp}.
Hence, $g$ is a generator of $\mathbb{Z}_p^*$.
We note that a similar proof of totality for the \DLP in $\mathbb{Z}_p^*$ was used by Krajíček and Pudlák~\cite{KrajicekP98}, who studied \DLP in terms of definability in bounded arithmetic.

Our first observation is a straightforward upper bound for \DLogp that follows by showing its inclusion in \PWPP.

\begin{lemma}\label{lemma:DLogp-in-PWPP}
\DLogp is reducible to \DLog. 
\end{lemma}
\begin{proof}
Given an instance $(p,p_1,\ldots,p_n,k_1\ldots,k_n,g,y)$ of \DLogp, we first fix a representation of $\Z_p^*$ by $[p-1]$.
Then, we construct the natural instance $(s,f,id,g',t)$ of \DLog, where
\begin{itemize}
    \item $s=p-1$,
    \item $f$ implements multiplication in $\Z_p^*$ w.r.t. the fixed representation of $\Z_p^*$,
    \item $id$ is the representation of $1\in\Z_p^*$ as an element of $[s]$,
    \item $g'$ is the representation of the given generator $g$ of $\Z_p^*$ as an element of $[s]$,
    \item $t$ is the representation of $y\in\Z_p^*$ as an element of $[s]$.
\end{itemize}
Now, we show that any solution to this instance $(s,f,id,g',t)$ of \DLog gives a solution to the original instance of $\DLog_p$. There are five types of solutions in \DLog:
\begin{enumerate}
    \item $a \in [s]$ such that $\I(a) = t.$ Since $f$ corresponds to a valid group operation, it holds that $\I(a)=g^a$. Hence, $g^a=t=y$ and $a \in [s]=\{0,\dots,p-2\}$ is a solution to the original instance of $\DLog_p$.  
    
    \item $a,b \in [s]$ such that $f_\GG(a,b) \geq s$. Since $f$ corresponds to a valid group operation, this case cannot happen. 
    
    \item $a,b \in [s]$ such that $a \neq b$ and $\I(a)=\I(b)$. Since $f$ corresponds to a valid group operation, we get that
    $g^a=\I(a)=\I(b)=g^b$.
    Suppose without loss of generality that $a>b$. Then, the previous relationship implies that $g^{a-b}=1$, where $a-b \neq 0$ and $a-b <p-1$. This would be a contradiction with the fact the $g$ is a generator of $\mathbb{Z}_p^*$. Hence, this case cannot happen.  
    
    \item $a,b \in [s]$ such that $a\neq b$ and $f_\GG(t,\I(a))=f_\GG(t,\I(b))$. Since $f$ corresponds to a valid group operation, the previous equality implies that $y\cdot g^a = y \cdot g^b$. 
    By cancelling $y$, we get that $g^a = g^b$ with $a \neq b$ and $a, b \in [s]$. For the same reason as in the previous case, this case cannot happen.  
    
    \item $a,b \in [s]$ such that $\I(a)=f_\GG(t,\I(b))$ and $\I(a-b \bmod s) \neq t$. Since $f$ corresponds to a valid group operation, we get that $g^a = y \cdot g^b$ and $g^{a-b} = g^{a-b \bmod s} \neq y$, which is impossible. Hence, this case cannot happen.
    \qedhere
\end{enumerate}
\end{proof}

Note that the proof of~\Cref{lemma:DLogp-in-PWPP} shows that the indexing function defined by taking the respective powers of $g$ is a bijection and any instance of \DLogp has a \emph{unique} solution.
Thus, there is  a stronger upper bound on the complexity of \DLogp in terms of containment in the class \TFUP, i.e., the subclass of \TFNP of total search problems with syntactically guaranteed unique solution for every instance.

\begin{corollary}\label{cor:DLogp-in-TFUP}
$\DLogp\in\TFUP$
\end{corollary}

We consider the existence of a reduction of an arbitrary instance of \Collision to a search problem with a unique solution for all instances such as \DLogp to be implausible.
Thus, we conjecture that \DLogp cannot be \PWPP-complete.

\subsection{\Blichfeldt}\label{sec:Blichfeldt}
Both our reductions establishing \PWPP-hardness of \DLog and \PPP-hardness of \Index result in instances that induce groupoids unlikely to satisfy the group axioms.
In other words, the resulting instances do not really correspond to \DLP in any group.
It is natural to ask whether this property is common to other \PWPP and \PPP hardness results.
In this section, we revisit the problem \Blichfeldt introduced in~\cite{SotirakiZZ18} and show that its \PPP-hardness can be established via a reduction that relies solely on the representation of the computational problem and does not use the type of solution corresponding to the ones guaranteed by the Blichfeldt's theorem.
A natural question is whether formalizations of other problems  from computational number theory, e.g., the computational versions of the Minkowski's theorem and the Dirichlet's approximation theorem defined in~\cite{BanJPPR19}, exhibit a similar phenomenon.

Below, we use the natural extension of the bit composition and decomposition functions (defined in~\Cref{sec:Preliminaries}) when applied to vectors.

\begin{definition}[\Blichfeldt]\label{def:problem_Blich}
The search problem \Blichfeldt is defined via the following relation of instances and solutions.
\begin{description}
    \item[Instance:] A triple $(\mathbf{B},s,\circuit{V})$, where $\mathbf{B} \in \mathbb{Z}^{n \times n}$ is an $n$-dimensional basis, $s$ $\in \mathbb{Z}^+$ is a size parameter such that  $s \ge \operatorname{det}(\mathcal{L}(\mathbf{B}))$, and $\circuit{V}$ is a Boolean circuit with $k=\lceil \log(s) \rceil$ inputs and $l$ outputs defining a set of vectors $S \subseteq \mathbb{Z}^n$ as $S = \left\{\bc \left(\circuit{V}\left(\bd\left(i\right)\right)\right), i \in [s] \right\}$.
    \item[Solution:] One of the following: 
    \begin{enumerate}
        \item distinct $u,v \in \{0,1\}^n$ such that $\circuit{V}(u) = \circuit{V}(v)$,
        \item a vector $x$ such that $x \in S \cap \mathcal{L}(\mathbf{B})$,
        \item distinct $x, y \in S$ such that $x - y \in \mathcal{L}(\mathbf{B})$.
    \end{enumerate} 
\end{description}
\end{definition}

In their work, \cite{SotirakiZZ18} showed that \Blichfeldt is \PPP-hard by a reduction from \Pigeon that relies on some nontrivial properties of $q$-ary lattices.
We show that this is unnecessary and give a more direct reduction that exploits the circuit $\circuit{V}$ in the definition of \Blichfeldt.
One particularly interesting property of our reduction is that it completely bypasses the solutions corresponding to the Blichfeldt's theorem in \Blichfeldt.
Specifically, all instances produced by our reduction are defined w.r.t. a fixed basis $\mathbf{B}$.

\begin{lemma}
\problem{Blichfeldt} is $\PPP$-hard.
\end{lemma}
\begin{proof}
We show a reduction from \problem{Pigeon} to \problem{Blichfeldt}. We start with an arbitrary instance $\circuit{C}: \left\{0,1\right\}^n \to \left\{0,1\right\}^n$ of \problem{Pigeon}. If $\circuit{C}(0^n) = 0^n$, then we output $0^n$ as a solution to this instance $\circuit{C}$ without invoking the \problem{Blichfeldt} oracle. Otherwise, we construct an  instance of \problem{Blichfeldt} as follows: 
\begin{itemize}
    \item We define $\m{B} = 2\cdot I_n$, i.e., the $n\times n$ diagonal matrix with $2$'s on its diagonal and $0$'s elsewhere. 
    \item We set $s = 2^n$. 
    \item We define the circuit $\circuit{V}: \{0,1\}^n \to \{0,1\}^n$ as follows:
    \[
    \circuit{V}(x) = 
\begin{cases}
\circuit{C}(x) &\mbox{if } \circuit{C}(x) \neq 0^n,\\
\circuit{C}(0^n) &\mbox{otherwise}.\\
\end{cases}
\]
\end{itemize} 

Note that $\bc$ maps any binary string $x=(x_1,\ldots, x_n)$ output by $\circuit{V}$ to an identical vector $(x_1,\ldots, x_n)^T$ in $\mathbb{Z}^n$.
In particular, all coordinates are either $0$ or $1$ for any vector from  the set $S$ defined by  $s$ and $\circuit{V}$ from the above instance of \problem{Blichfeldt}.

We now show that any solution to the above \problem{Blichfeldt} instance gives a solution to the original \problem{Pigeon} instance $\circuit{C}$. First, notice that $\operatorname{det}(\m{B}) = 2^n \ge s= 2^n$. Thus, we defined a valid instance of \Blichfeldt w.r.t. \Cref{def:problem_Blich}. 
Next, we argue that there are no solutions of the second type in the above instance, i.e., the \problem{Blichfeldt} oracle cannot return a vector $x$ such that $x \in S \cap \mathcal{L}(\textbf{B})$: From the definition of $\circuit{V}$ and the fact that $\circuit{C}(0^n) \neq 0^n$, we get that $0^n \notin S$, but $0^n$ is the only vector in $\{0,1\}^n \cap \mathcal{L}(\m{B})$.  

Furthermore, there are also no solutions of the third type for any instance defined as above, i.e., the \problem{Blichfeldt} oracle cannot return distinct vectors $x, y \in S$ such that $x - y \in \mathcal{L}(\m{B})$.
Indeed, all vectors in $S$ have coefficients in $\left\{0,1\right\}$ and, thus, all coefficients of $x - y$ would lie in $\left\{-1,0,1\right\}$. However, the only such vector also contained in $\mathcal{L}(\m{B})$ is $0^n$, which would imply $x = y$.

Hence, there are only solutions of the first type and the \problem{Blichfeldt} oracle returns two distinct strings $u,v \in \{0,1\}^n$ such that $\circuit{V}(u) = \circuit{V}(v)$. If $\circuit{C}(u) = 0^n$, then $u$ is a solution to the original \problem{Pigeon} instance $\circuit{C}$. Similarly, if $\circuit{C}(v) = 0^n$, then $v$ is a solution to $\circuit{C}$. Otherwise, it holds that $\circuit{C}(u) \neq 0^n \neq \circuit{C}(v)$. Hence, from the definition of $\circuit{V}$, we get that 

$$\circuit{C}(u) = \circuit{V}(u) = \circuit{V}(v) = \circuit{C}(v),$$

and the pair $u, v$ is a solution to the original instance $\circuit{C}$ of \Pigeon. 
\end{proof}

\section*{Acknowledgements}
We wish to thank the anonymous reviewers for helpful suggestions on the presentation of our results.
The first author is grateful to Chethan Kamath for multiple enlightening discussions about \TFNP and for proposing \Dove as a name for a total search problem contained in the complexity class \PPP.

\bibliographystyle{alpha}
\bibliography{bibliography}

\end{document}